\documentclass[12pt]{article}

\usepackage{amsmath}
\usepackage{amsfonts}
\usepackage{amssymb}
\usepackage{latexsym}
\usepackage[dvips]{graphics}
\usepackage[utf8]{inputenc}
\usepackage{mathrsfs}
\usepackage{epsfig}
\usepackage[footnotesize, SL]{subfigure}
\usepackage[small,bf]{caption}	
\usepackage{xy}
\numberwithin{equation}{section}
\newtheorem{proposition}{Proposition}[section]

\newtheorem{theorem}[proposition]{Theorem}

\newtheorem{lemma}[proposition]{Lemma}

\renewcommand{\to}{\longrightarrow}
\renewcommand\P{\mathbb{P}}

\newcommand{\B}{\mathcal{B}}
\renewcommand{\L}{\mathcal{L}}	
\newcommand{\Y}{\mathcal{Y}}
\newcommand{\N}{\mathbb{N}}

\newcommand{\Z}{\mathcal{Z}}

\newcommand{\Np}{N_+}
\newcommand{\Nm}{N_-}
\newcommand{\phiinv}{\phi^{-1}}

\newenvironment{proof}[1][Proof]{\begin{trivlist}
\item[\hskip \labelsep {\bfseries #1}]}{\end{trivlist}}
\newcommand{\qed}{\begin{flushright} $\square$
\end{flushright}}

\begin{document}
 
 
\topmargin 0pt
\oddsidemargin 5mm
\headheight 0pt
\topskip 0mm

\addtolength{\baselineskip}{0.3\baselineskip}
     
\pagestyle{empty}

\hfill

     \vspace{0.8cm}
     
     \begin{center}

     {\Large \bf Condensation in nongeneric trees}

     \medskip
     \vspace{0.8 truecm}
     
     {\large \bf \today}
     
     \vspace{0.8truecm}

{\bf Thordur Jonsson$^{1}$ and Sigurdur Örn Stefánsson$^{1,2}$}

\vspace{0.4 truecm}

{$^1$ The Science Institute, University of Iceland

Dunhaga 3, 107 Reykjavik,

Iceland

}

\vspace{0.2 truecm}

{$^2$ NORDITA,

Roslagstullsbacken 23, SE-106 91 Stockholm,

Sweden

}

\vspace{0.4 truecm}
\texttt{
thjons@raunvis.hi.is \\
sigurdurorn@raunvis.hi.is
}
\vspace{.3 truecm}

\end{center}
  
\noindent {\bf Abstract.}  We study nongeneric planar trees and prove the existence
of a Gibbs measure on infinite trees obtained as a weak limit of the
finite volume measures.  It is shown that in the infinite volume limit
there arises exactly one vertex of infinite degree and the rest of the
tree is distributed like a subcritical Galton-Watson tree with mean
offspring probability $m<1$.   We calculate the rate of
divergence of the degree of the highest order vertex of finite trees
in the thermodynamic limit and show it goes like $(1-m)N$ where $N$ is
the size of the tree.   These trees have infinite spectral dimension with
probability one but the
spectral dimension calculated from the ensemble average of the
generating function for return probabilities is given by
$2\beta -2$ if the weight $w_n$ of a vertex of degree $n$ is
asymptotic to $n^{-\beta}$.

\newpage
\pagestyle{plain}
    
\section {Introduction}

In the recent past the interest of scientists in various classes of
random graphs and networks has increased dramatically due to the 
many applications of these mathematical structures to describe objects
and relationships in subjects ranging from pure mathematics and
computer science to physics, chemistry and biology.  An important
class of graphs in this context are tree graphs, both because many
naturally appearing random graphs are trees and also because trees are
analytically more tractable than general
graphs and one expects that
some features of general random graphs can be understood by looking
first at trees.

In this paper we study an equilibrium statistical mechanical model of
planar trees.  The parameters of the model are given by a sequence of non--negative numbers $(w_i)_{i\geq 1}$, referred to as branching weights. To a finite tree $\tau$ we assign a Boltzmann weight
$$
W(\tau )=\prod_{i\in V(\tau)} w_{\sigma_i}
$$
where $V(\tau)$ is the vertex set of $\tau$ and $\sigma_i$ is the degree of the vertex $i$. The model is local, in the sense that the energy of a tree is given by the sum over the energies of individual vertices. In~\cite{bialas:1996ya} the critical
exponents of this model were calculated and its phase structure was
described. It was argued that the
model exhibits two phases in the thermodynamic limit: a fluid
(elongated, generic) phase where the
trees are of a large diameter and have vertices of finite degree and a
condensed (crumpled) phase
where the trees are short and bushy with exactly one vertex of
infinite degree.

A complete characterization of the fluid phase, referred to as generic
trees, was given
in \cite{durhuus:2003,durhuus:2007} where it was shown that the Gibbs
measures
converge to a measure concentrated on the set of trees with exactly
one non--backtracking path from the root to infinity having critical
Galton-Watson outgrowths. In~\cite{durhuus:2007} it was
furthermore
proved that the trees have
Hausdorff dimension $d_H = 2$ and spectral dimension $d_s = 4/3$ with
respect to the infinite volume measure. The purpose of this paper is
to establish
analogous results for the condensed phase. Preliminary results in this
direction
were obtained in~\cite{sigurdur:2007}.

One of the motivations for the study of the tree model is that a
similar phase structure
is seen for more
general class of graphs in models of simplicial
gravity~\cite{agishtein:1992,ambjorn:1992}. In
these models the elongated phase seems to be effectively described by
trees~\cite{ambjorn:1995} and it has been
established by numerical methods that in the condensed phase a single
large simplex appears whose size
increases linearly with the graph
volume~\cite{catterall:1996,hotta:1996}.
In~\cite{bialas:1997:495} it was proposed that the same mechanism is
behind the phase
transition in the different models and
the so called constrained mean field model was introduced in order to
capture this feature. This idea was developed in a series of
papers~\cite{bialas:2000,bialas:1998:416,bialas:1997qs,bialas:1998:63,
bialas:1999}
where the model was studied under the name ``balls in
boxes'' or ``backgammon model''. The model
consists of placing
$N$ balls into $M$ boxes and assigning a
weight to each box depending only on the number of balls it contains.
In~\cite{bialas:1997qs} the critical exponents were
calculated and the two phases characterized. The distribution of the
box occupancy number was derived
and it was argued that in the condensed phase exactly one box contains
a large number of balls which
increases linearly with the system size.

A model equivelant to the ``balls in boxes'' model was
studied in a recent paper~\cite{jonsson:2009}. It is an equilibrium
statistical mechanical model with a local action of the form described
above but the class of trees is
restricted to so called caterpillar graphs. Caterpillars are trees
which have the property that if all
vertices of degree one and the edges containing them are removed, the
resulting graph is linear. The
caterpillar model was solved by proving convergence of the Gibbs
measures to a measure on infinite
graphs and the limiting measure was completely characterized. It was
shown that in the fluid phase the
measure is concentrated on the set of caterpillars of infinite length
and that in the condensed
phase it is
concentrated on the set of caterpillars which are of finite length and
have precisely one vertex of
infinite degree. This was the first rigourous treatment of the
condensed phase in models of the above
type. A model of random combs, equivalent to the caterpillar model was
studied
in~\cite{durhuus:2009} where analagous results were obtained for the
limiting measure.
A closely related phenomenom of condensation also appears in dynamical
systems such as
the zero range process, see e.g.\ \cite{evans:2005}.

In this paper we use techniques similar to those of \cite{jonsson:2009}
with some additional input from probability theory to prove convergence of the Gibbs measures in the condensed phase of
the planar tree model.  In Section 2
we generalize the definition of planar trees to
allow for vertices
of infinite degree and define a metric on the set of planar trees
which has the nice
properties that the metric space is compact and that
the subset of finite trees is dense. 
In Section 3 we recall the definition of generic and nongeneric trees,
define the partition functions of interest and recall the relation to
Galton-Watson processes.
Section 4 is the technical core of the paper.  There we review
some results we need from probability theory and then show that the 
partition function $Z_N$ for nongeneric 
trees of size $N$ has the asymptotic behavior
\begin{equation}
Z_N\sim N^{-\beta}\zeta_0^{-N}
\end{equation}
where $\zeta_0$ is a constant and $\beta$ is the exponent of 
the power decay of the
weight $w_n$ of vertices of degree $n$, i.e.\ $w_n\sim n^{-\beta}$.
In Section 5 we establish the
existence of the infinite volume Gibbs measure and prove that in the condensed
phase it is concentrated on the set of trees of finite diameter with
precisely one vertex of infinite degree and that the rest of the tree is distributed as a
subcritical Galton-Watson process with mean offspring probability $m<1$. We prove that
for finite trees the degree of the large vertex grows linearly with the system size $N$
as $(1-m)N$ with high probability, confirming the result stated
in~\cite{burda:2001dx}. We conclude in Section 6 
by calculating the annealed spectral dimension of the trees in the
condensed phase.
In~\cite{correia:1997gf} it was claimed, on the basis of scaling
arguments, that the
spectral dimension is $d_s = 2$. We prove, however, that if the spectral
dimension exists
it is given by $d_s=2\beta -2$.
In fact, it takes the same value as the spectral dimension of the
condensed phase in the
caterpillar model~\cite{jonsson:2009}.  

\section {Rooted planar trees}
In this section we recall the definition of rooted planar trees and
define a convenient 
metric on the set of all such trees.  We establish some
elementary properties of the trees as a metric space which will be
needed in the construction of a measure on infinite trees.  The combinatorial 
definition of planar trees below is in the spirit 
of~\cite{durhuus:2003} with the change that we allow for 
vertices of infinite degree. 
We require this extension since vertices of infinite degree appear in the
thermodynamic limit in the nongeneric 
phase of the random tree model in Section \ref{s:nongeneric}.  

The 
planarity condition means that links incident on a vertex are
cyclically ordered.  When the 
degree of a vertex is infinite there are nontrivial
different possibilities of ordering 
the links and therefore the planarity condition must by carefully
defined.  We allow vertices of at 
most countably infinite degree and the edges are given the simplest
possible ordering, i.e.\ if we look at the set of edges leading away from
the root at a given vertex, then the smallest edge is the leftmost one
which is required to exist and the remaining edges are ordered as $\mathbb{N}$.  Note that we could just as
well choose the rightmost edge as the smallest and order the remaining
ones counterclockwise as $\mathbb{N}$. The root will always be taken to have order one for convenience but
this is not an essential assumption.

Let $(D_R)_{R\geq 0}$ be a sequence of pairwise disjoint, countable sets with the properties that if 
 $D_R = \emptyset$ then $D_r = \emptyset$ for all $r\geq R$. The sets $D_0$ and $D_1$ are defined to have only a single element.  The set
 $D_R$ will eventually denote the set of vertices at graph distance $R$ from the
 root.  We will denote the 
 number of elements in a set $A$ by $|A|$. To introduce the edges and the planarity condition, we define orderings on each of the sets $D_R$ and order preserving maps 
\begin {equation}
\phi_R : D_R \longrightarrow D_{R-1}, \quad \quad R \geq 1
\end {equation}
which satisfy the following: For each vertex $v \in D_{R-1}$ such that \\ $|\phiinv_R(v)| = \infty$, there exists an order preserving isomorphism 
\begin {equation}
\psi_v : \mathbb{N} \longrightarrow \phiinv_R(v).
\end {equation}
In this notation $\phi_R (w)$ is the parent of the vertex $w$ and
$|\phiinv_R(v)| +1$ is the degree of the vertex $v$, denoted $\sigma_v$.
One can show by induction on $R$ that such orderings on the sets 
$D_R$ can be defined 
and that they are well orderings, i.e.\ each subset of $D_R$ has a smallest
element.  It is not hard to check that given the ordered sets $D_R, 
R\geq 0$, and the order preserving maps $\phi_R, R\geq 1$ 
with the above properties 
the maps $\psi_v$ are unique.  For a vertex $v$ of a finite degree we
can also define the mappings $\psi_v$ and they are trivial. 

Let $\tilde{\Gamma}$ be the set of all pairs of 
sequences $\{(D_0,D_1,D_2,\ldots),(\phi_1,\phi_2,\ldots)\}$ which satisfy the 
above conditions.
We define an equivalence relation $\sim$ on $\tilde{\Gamma}$ by 
\begin{equation}
\{(D_0,D_1,\ldots),(\phi_1,\phi_2,\ldots)\} \sim 
\{(D'_0,D'_1,\ldots),
(\phi'_1,\phi'_2,\ldots)\}
\end{equation}
if and only if for all $R\geq 1$ there exist order 
isomorphisms $\chi_R : D_R \longrightarrow D'_R$ 
such that $\phi'_R = \chi_{R-1} 
\circ \phi_R \circ \chi_R^{-1}$. Note that since the sets $D_R$ are well ordered for all $R\geq 1$ 
the order isomorphisms $\chi_R$  are unique. Define $\Gamma = \tilde{\Gamma}/\sim$. 
If $\tau \in \tilde{\Gamma}$ we denote the equivalence class of $\tau$ by 
$[\tau]$ and call it a rooted planar tree.  
As a graph, the tree has a vertex set 
\begin{equation}
V = \bigcup_{R=0}^\infty D_R
\end{equation}
and an edge set
\begin{equation}
E = \{(v,\phi_R(v)) ~|~ v \in D_R, R\geq 1\}
\end{equation}
which are independent of the representative $\{(D_0,D_1,\ldots),(\phi_1,
\phi_2,\ldots)\}$ up to graph isomorphisms. The 
single element in $D_0$ is called 
the root and denoted by $r$. 
We denote the set of all rooted planar trees on $N$ 
edges by $\Gamma_N$ and the set of finite rooted planar trees by 
$\Gamma' = \bigcup_{N=1}^\infty \Gamma_N$.

In the following, all properties of trees $[\tau] \in \Gamma$ that we are 
interested in are independent of representatives and we write $\tau$ instead 
of $[\tau]$.  
We then write $D_R(\tau)$, $\phi_R(\cdot,\tau)$, 
$\psi_v(\cdot,\tau)$, $\sigma_v(\tau)$ etc.\ when we need the more detailed 
information on $\tau$. If it is clear which tree we are working with we skip 
the argument $\tau$. When we draw the trees in the plane we use the convention that $\psi_v(k)$ is the 
$k$-th vertex clockwise from the nearest neighbour of $v$ closest to
the root, see Figure \ref{f:rodun}. 
\begin{figure} [!t]
\centerline{\scalebox{0.8}{\includegraphics{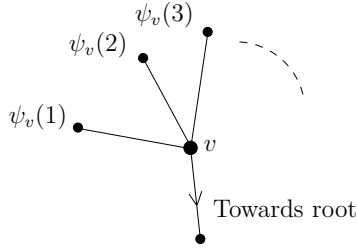}}}
\caption{The ordering of $\phiinv_R(v)$.} \label{f:rodun}
\end{figure}

For a tree $\tau\in\Gamma$ we denote its height, i.e.~the 
maximal graph distance 
of its vertices from the root, by $h(\tau)$. For a pair of vertices
$v$ and $w$ we denote 
the unique shortest path from $v$ to $w$ by $(v,w)$.  The ball of 
radius $R$, $B_R(\tau)$ is defined as the subtree of $\tau$ generated by 
$D_ 0(\tau),D_1(\tau)\ldots D_R(\tau)$.

We define the {\it left ball} of radius $R$, $L_R(\tau)$, 
as the subtree of $B_R(\tau)$ generated by subsets 
$E_S \subseteq D_S(B_R(\tau))$, $S=0,...R$, such that 
$E_0 = D_0(B_R(\tau))$, $E_1=D_1(B_R(\tau ))$ and
\begin {equation}
E_{S} = \{\psi_v(i) ~|~ v \in E_{S-1}, i=1,2,\ldots,\min\{R,\sigma_v\}-1\},
\end {equation}
see Fig. \ref{f:def}.  We denote the number of edges in a tree $\tau$ 
by $|\tau|$. It is easy to check that for all $\tau \in \Gamma$
\begin {equation}  \label{Lsize}
 |L_R(\tau)| \leq {\frac {(R-1)^{R}-1}{R-2}},
\end {equation}
whereas the number of elements in $B_R(\tau )$ can be infinite. 
\begin{figure} [!t]
\centerline{\scalebox{0.85}{\includegraphics{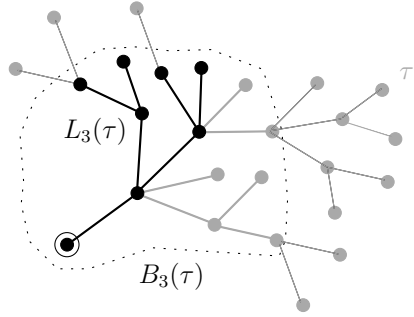}}}
\caption{An example of the subgraphs $B_R(\tau)$ and $L_R(\tau)$.} \label{f:def}
\end{figure}
We define a metric $d$ on $\Gamma$ by
\begin {equation}
d(\tau,\tau') = \min\left\{\left.\frac{1}{R} ~\right|~ L_R(\tau) = 
L_R(\tau'), ~R \in \N\right\}, \quad \quad \tau,\tau'\in\Gamma.
\end {equation}
It is elementary to check that this is in fact a metric. Note that if we allow any ordering on the infinite sets $|\phiinv_R(v)|$, but still insist that they have a smallest element, then this ordering is in general not a a well ordering and $d$ is only a pseudometric.

Denote the open ball in $\Gamma$ centered at $\tau_0$ and with radius $r$ by
\begin {equation}
 \B_r(\tau_0) = \{\tau \in \Gamma ~|~ d(\tau_0,\tau) < r\}.
\end {equation}
\begin {proposition} \label{p:openclosed}
 For $r>0$  and $\tau_0 \in \Gamma$, the ball $\B_r(\tau_0)$ is both open and 
 closed and if $\tau_1 \in \B_r(\tau_0)$ then $\B_r(\tau_1) = \B_r(\tau_0)$.
\end {proposition}
\begin {proof}
It is easy to see that open balls are also closed 
since the possible positive values of $d$ form a discrete set. To prove 
the second statement take a ball $\B_r(\tau_0)$ and a tree $\tau_1 \in \B_r(\tau_0)$. 
First take an element $\tau_2 \in \B_r(\tau_1)$. We know that 
$L_R(\tau_1) = L_R(\tau_0)$ and $L_R(\tau_1) = L_R(\tau_2)$ for all 
$R < 1/r$ so obviously $L_R(\tau_0) = L_R(\tau_2)$ for all $R < 1/r$. Therefore 
\begin {eqnarray}
d(\tau_2,\tau_0) &\leq& \min\left\{\left.\frac{1}{R} ~\right|~ L_R(\tau_2) = 
L_R(\tau_0),~R < 1/r + 1\right\} \nonumber\\
&=& \frac{1}{\lfloor1/r+1\rfloor} < r.
\end {eqnarray}
Therefore $\tau_2 \in \B_r(\tau_0)$ and thus $\B_r(\tau_1) \subseteq \B_r(\tau_0)$. 
With exactly the same argument we see that $\B_r(\tau_0) \subseteq \B_r(\tau_1)$ 
and therefore the equality is established. \qed
\end {proof}

\begin {proposition} \label{p:compact}
The metric space $(\Gamma,d)$ is compact and 
the set $\Gamma'$ of finite trees is a countable dense subset of $\Gamma$.
\end {proposition}
\begin {proof}
To prove compactness it is enough to note that by (\ref{Lsize}), for each $R \in \N$, 
the set  $\{L_R(\tau)~|~\tau\in\Gamma\}$ is finite. The result then follows by the same arguments applied to the set of random walks in \cite{durhuus:2003}.

In order to prove the density of $\Gamma '$
we note that the sequence 
$\left(L_n(\tau)\right)_{n\in\N}$ is in $\Gamma'$ 
by (\ref{Lsize})
and clearly
converges to $\tau$.   \qed
\end {proof}


\section {Generic and nongeneric trees} \label {s:themodel}

In this section we define the tree ensemble that we study 
and discuss some of its
elementary properties.
Let $w_n$, $n\geq 1$ be a sequence of non--negative numbers which we call 
{\it branching weights}. For technical convenience we will always take  
\begin {equation} \label {statedassumption}
w_1,w_2 > 0 \quad \quad \text{and} \quad \quad w_n > 0  
\quad \quad \text{for some $n \geq 3$}.
\end {equation}
Let $V(\tau )$ be the set of vertices in $\tau$.
The finite volume partition function is defined as
\begin {equation}
Z_N = \sum_{\tau\in\Gamma_N} ~\prod_{i\in V(\tau)\setminus\{r\}} w_{\sigma_i}
\end {equation}
where $\sigma_i$ is the degree of vertex $i$. We define a probability distribution $\nu_N$ on $\Gamma_N$ by
\begin {equation}
 \nu_N(\tau) = Z_N^{-1} \prod_{i\in V(\tau)\setminus\{r\}} w_{\sigma_i}.
\end {equation}

The weights $w_n$, or alternatively the measures $\nu_N$, define a tree 
ensemble. Note that $\nu_N$ is not affected by a rescaling of the branching 
weights of the form $w_n \rightarrow w_n a b^n$ where $a,b>0$. We introduce 
the generating functions 
\begin {equation}\label{grandcanonical}
\Z(\zeta) = \sum_{N=1}^\infty Z_N \zeta^N
\end {equation}
and
\begin {equation}\label{potential}
g(z) = \sum_{n=0}^\infty w_{n+1}z^n.
\end {equation}
Then we have the standard relation
\begin {equation} \label {eqtree}
\Z(\zeta) = \zeta g(\Z(\zeta))
\end {equation}
which is explained in Fig.~\ref{f:eqtree}.
\begin {figure} [!b] 
\begin {center}
\scalebox{1}{\includegraphics{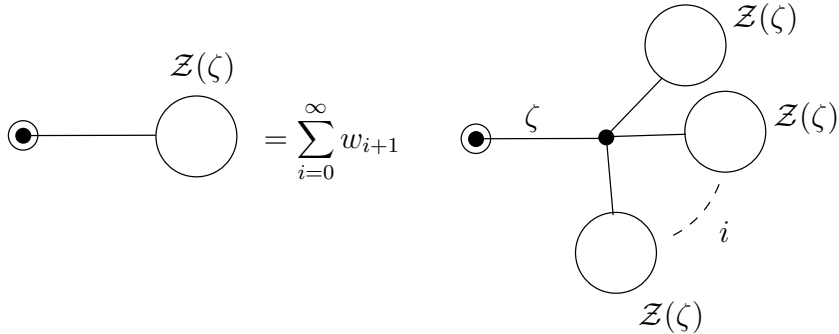}}
\caption{A diagram explaining the recursion (\ref{eqtree}). The root is 
indicated by a circled point.} \label{f:eqtree}
\end {center}
\end {figure}

We denote the radius of convergence of $\Z(\zeta)$ and $g(z)$ by $\zeta_0$ 
and $\rho$, respectively, and define $\Z_0 = \Z(\zeta_0)$. If $\Z_0 <
\rho$, then by definition 
we have a {\it generic ensemble} of trees \cite{durhuus:2007}. 
Otherwise we say that we 
have a {\it nongeneric ensemble}. If $\rho = \infty$ we always have a generic 
ensemble. If $\rho$ is finite we can assume that $\rho=1$ by scaling the 
branching weights $w_n \rightarrow w_n \rho^{n-1}$. 

There is a useful relation between the tree ensemble 
$(\Gamma_N,\nu_N)$ and Galton--Watson (GW) trees (see e.g.~\cite{Harris:1963}). 
Let $p_n$, $n=0,1,2 \ldots$, be the offspring probability distribution
for a GW tree.   Then we link a vertex of order one (the
root) 
to the ancestor of the GW tree and obtain a rooted tree with a root of
degree 1. The GW process gives rise to a probability 
measure on the set of all finite trees
\begin {equation} \label{GWmeasure}
\mu(\tau) = \prod_{i \in V(\tau) \setminus \{r\}} p_{\sigma_i - 1}, 
\quad\quad\quad \text{where} \quad\quad \tau \in \Gamma'.
\end {equation}
Let $m$ be the average number of offsprings in the GW process. 
If $m>1$ the process is said to be supercritical and the probability that it  
survives forever is positive. If $m = 1$ the process is said to be critical 
and it dies out
with probability one. If $m < 1$ the process is said 
to be subcritical and it dies out exponentially fast \cite{Harris:1963}.

The probability distribution $\nu_N$ can 
be obtained from a GW process 
with offspring probabilities 
\begin {equation} \label{GWprob}
p_n = \zeta_0 w_{n+1}\Z_0^{n-1}
\end {equation}
by conditioning the trees to be of size $N$
\begin {equation}
 \nu_N(\tau) = \frac{\mu(\tau)}{\mu(\Gamma_N)}.
\end {equation}
The mean offspring probability is then
\begin {equation} \label{meanoffspring}
 m = \Z_0 \frac{g'(\Z_0)}{g(\Z_0)}.
\end {equation}
Generic trees always correspond to critical GW processes \cite{durhuus:2007}
and nongeneric trees can correspond to either critical or 
subcritical GW processes. In all cases $m \leq 1$. We now analyze this in more detail. 

Fix a set of branching weights $w_n$ 
which give $\rho=1$ but let $w_1$ be a free parameter 
of the model
which at this stage can be either generic or nongeneric.
Define   
\begin{equation}
h(\Z) = \frac{g(\Z)}{\Z}.
\end{equation}
From (\ref{eqtree}) we see that $h(\Z) = 1/\zeta(\Z)$ for $\Z \leq \Z_0$. 
Differentiating $h$ we get
\begin{equation}
h'(\Z) = \frac{g(\Z)}{\Z^2}\left[\Z\frac{g'(\Z)}{g(\Z)} - 1\right]
\end{equation}
and again
\begin{equation}
h''(\Z) = \frac{g''(\Z)}{\Z}-\frac{2}{\Z}h'(\Z).
\end{equation} 
The genericity condition means that $h$ has a quadratic minimum at 
$\Z = \Z_0 < 1$, see Fig. \ref{phasediagram1}.  It follows that 
$m = 1$, showing that the generic phase corresponds to critical GW 
trees. Furthermore, given $\Z_0<1$ and the branching weights 
$w_n$, $n\geq 2$, we have $w_1 = 
\sum_{n=2}^{\infty}(n-2)w_n \Z_0^{n-1}$. We can therefore clearly 
make any model with $\rho = 1$ generic by choosing
\begin{equation}
 w_1< \sum_{n=2}^{\infty}(n-2)w_n \equiv w_c.
\end{equation} 
Here $w_c$ is a critical value for $w_1$ which depends 
on $w_n$ for $n\geq 3$.  We note that if $w_c = \infty$, 
i.e.~if $g'(z)$ diverges as $z\rightarrow 1$, we always have a generic ensemble.
\newline
\begin {figure} [!t] 
\begin {center}
\scalebox{0.7}{\includegraphics{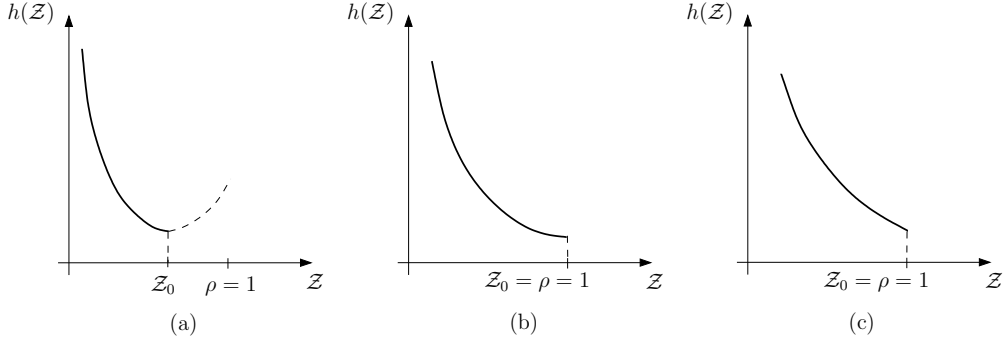}}
\caption{The three possible scenarios. a) Generic, critical, $w_1 < w_c$. \newline b) nongeneric, critical, $w_1 = w_c$. c) nongeneric, subcritical, $w_1 > w_c$.} \label{phasediagram1}
\end {center}
\end {figure}

The next possible scenario is that $h$ has a quadratic minimum at $\Z = \Z_0 = 1$. 
This happens when $w_1 = w_c$ or in other words when $m = 1$. 
This is a nongeneric ensemble which still corresponds to critical GW trees. 

Finally, by choosing $w_1 > w_c$, $h$ has no quadratic minimum and 
$m < 1$. In this case the trees are nongeneric 
and correspond to subcritical GW trees as we will explore
in detail in the next section.

\section{Subcritical nongeneric trees} \label{s:nongeneric}

In this section we examine the subcritical nongeneric phase and determine 
the asymptotic behaviour of $Z_N$.  This will allow us to construct
the infinite volume Gibbs measure in the next section.  

We fix a number $\beta\geq 0$ and for $n\geq 2$ we 
choose the branching weights such that 
\begin {equation} \label{bw}
w_n = n^{-\beta}(1+o(1)),\quad n\geq 2
\end {equation}
and let $w_1$ be a free parameter. In this case $\rho = 1$. If $\beta \leq 2$ 
then $g'(1) = \infty$ and therefore we have the generic phase for all values 
of $w_1$. If $\beta > 2$ we can have any one of the three cases discussed in the 
previous section depending on the value of $w_1$, see Fig. \ref{phasediagram2}.
\begin {figure} [!t] 
\begin {center}
\scalebox{0.9}{\includegraphics{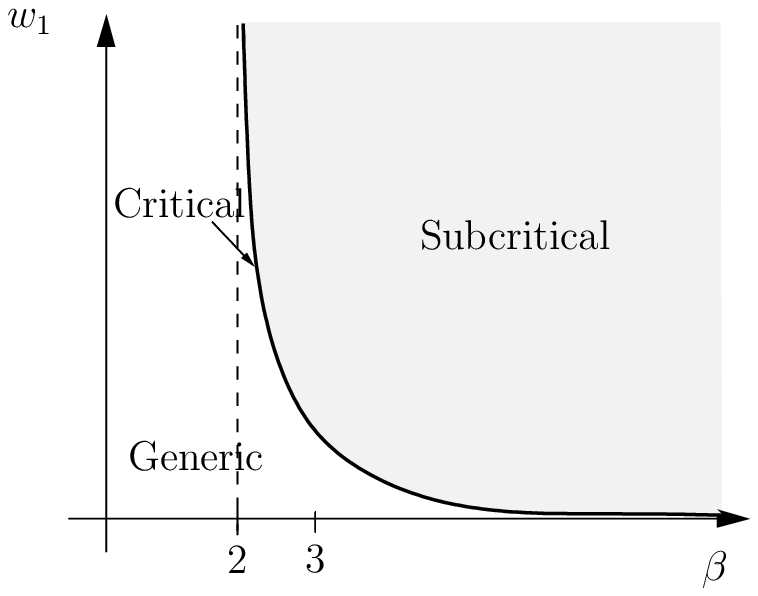}}
\caption{A diagram showing the possible phases of the trees. The critical 
line is determined by the equation $w_1 = w_c$.} \label{phasediagram2}
\end {center}
\end {figure}
Now choose $\beta >2$ and $w_1 > w_c$ such that
\begin {equation} \label{ngcondition}
m = \frac{g'(1)}{g(1)} < 1
\end {equation}
so we are in the nongeneric, subcritical phase. Then $\Z_0 = \rho = 1$ and we 
see from (\ref{eqtree}) that

\begin {equation}
 \zeta_0 = \frac{1}{g(1)}.
\end {equation}
The main result of this section is the following.

\begin {theorem} \label{th1}
If the branching weights (\ref{bw}) satisfy (\ref{ngcondition}) then the
partition function has the asymptotic behaviour
\begin {equation}
 Z_{N} = (1-m)^{-\beta} N^{-\beta} \zeta_0^{1-N}\left(1+o(1)\right).
\end {equation}
\end {theorem}

The remainder of this section is devoted to a proof of this theorem.
To determine the large $N$ behaviour of $Z_N$ we split it into two parts,

\begin {equation} \label{splitsum}
Z_N = Z_{1,N} + E_N,
\end {equation}
where $Z_{1,N}$ is the contribution to $Z_N$ from trees which have exactly $1$ 
vertex of maximal degree 
and $E_N$ is the contribution to $Z_N$ from trees which have $\geq$ 2 vertices 
of maximal degree.  We will estimate these two terms seperately and show 
that for large $N$ the main contribution comes from $Z_{1,N}$. It
follows 
from the proof that large trees, of size $N$, are most likely to have exactly 
one large vertex which is approximately of degree $(1-m)N$. This will be 
stated more precisely in Section \ref{s:measure}.  The arguments used in the 
proof of Theorem \ref{th1} rely on a ``truncation method'' and some classical 
results from probability theory. We begin the proof by defining truncated 
versions of the generating functions introduced in the previous section. Then 
we introduce some notation and terminology 
from probability theory and state a few lemmas. 
In Subsection \ref{ss:Z1N} we analyse the asymptotic behaviour of $Z_{1,N}$ 
and in Subsection \ref{ss:EN} we do the same for $E_N$. 

For the truncation method, we will need the following definitions. Let $L_{i,N}$ 
be the finite volume partition function for trees on $N$ edges which have all 
vertices of degree $\leq i$ and define the generating functions
\begin {equation}
 \L_i(\zeta) = \sum_{N=1}^\infty L_{i,N}\zeta^N
\end {equation}
and
\begin {equation}
\ell_i(z) = \sum_{n=0}^{i-1} w_{n+1}z^n. 
\end {equation}
We have the standard relation
\begin {equation} \label {eqtreelower}
\L_i(\zeta) = \zeta \ell_i(\L_i(\zeta))
\end {equation}
obtained in the same way as (\ref{eqtree}). Let  $Y_{j,i,N}$  be the finite 
volume partition function for trees on $N$ edges which have all vertices of 
degree $\leq i$ and one marked (but not weighted) vertex of degree one at 
distance $j$ from the root.  Define 
\begin {equation}
 \Y_{j,i}(\zeta) = \sum_{N=1}^\infty Y_{j,i,N}\zeta^N
\end {equation}
and
\begin {equation}
 \Y_{i}(\zeta) = \sum_{j=1}^\infty \Y_{j,i}(\zeta).
\end {equation}
With generating function arguments we find that
\begin {equation} \label{twopointlength}
 \Y_{j,i}(\zeta) = \zeta \ell_i'(\L_i(\zeta)) \Y_{j-1,i}(\zeta) 
\end {equation}
for $j\geq 2$, see Fig.~\ref{marked}. Using $\Y_{1,i}(\zeta) = \zeta$ this 
yields by induction 
\begin {equation} 
 \Y_{j,i}(\zeta) = \zeta \Big(\zeta \ell_i'(\L_i(\zeta))\Big)^{j-1}.
\end {equation}
\begin {figure} [!t] 
\begin {center}
\scalebox{0.9}{\includegraphics{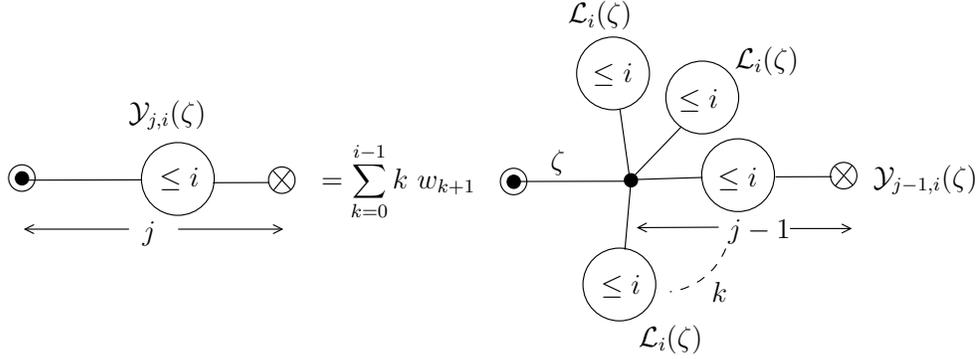}}
\caption{A diagram explaining  (\ref{twopointlength}). The marked vertex is 
indicated by $\otimes$. The balloons containing ``$\leq i$'' are 
trees which have vertices of degree at most $i$. If the degree of the nearest 
neighbour to the root is $k+1$, there are $k$ different ways of placing the 
marked vertex in a balloon.} \label{marked}
\end {center}
\end {figure}
Summing over $j$ we get
\begin {equation} \label{twopoint}
\Y_{i}(\zeta) = \frac{\zeta}{1-\zeta\ell_i'(\L_i(\zeta))}.
\end {equation}

\subsection{Tools from probability theory}

It will be useful to formulate our problem in probabilistic language. 
Define the probability generating functions
\begin {equation}
f_i(z) = \frac{\ell_i(z)}{\ell_i(1)} \quad \quad \text{and} 
\quad \quad f(z) = \frac{g(z)}{g(1)}.
\end {equation}

If $A$ is an event, we let $\P (A)$
denote the probability of $A$. Let $X^{(i)}_{1},X^{(i)}_{2},\ldots$ be i.i.d.~random variables which have a probability generating function $f_i(z)$, i.e.~
\begin{equation}
\mathbb{P}(X^{(i)}_{j} = k) = \left\{\begin {array}{ll} 
w_{k+1}/\ell_i(1) & \text{if $0\leq k\leq i-1$}, \\
0 & \text{if $k > i-1,$}
\end {array}\right.
\end{equation} 
and let $X_{1},X_{2},\ldots$ be i.i.d.~random variables which have a probability generating function $f(z)$.  Define 
\begin {equation}
m_i = \mathbb{E}(X^{(i)}_{j}),\quad V_i = \text{Var}(X^{(i)}_{j}), 
\quad \quad S^{(i)}_{N} = X^{(i)}_{1}+\ldots+X^{(i)}_{N}
\end {equation}
and
\begin {equation} \label {rvdef}
\quad \quad  S_{N} = X_{1}+\ldots+X_{N}.
\end {equation}
Note that $m=\mathbb{E}(X_j)$ and from (\ref{ngcondition}) we know that 
$m < 1$. Clearly $m_i \longrightarrow m$ as $i \longrightarrow \infty$. 
We need now a few lemmas, the first three deal with convergence rates 
in the weak law of large numbers.  

\begin {lemma} \label{lemmabk}
For any $\epsilon > 0$ and any $s < \beta - 2$ we have
\begin {equation}
 \lim_{N\rightarrow\infty}N^s\P\left(\left|\frac{S_N}{N}-m\right| 
 > \epsilon\right) = 0.
\end {equation}
\end {lemma}
\begin {proof}
It is clear that $\mathbb{E}(|X_j|^t) < \infty$ for all $t<\beta - 1$
and the same is true for the translated random variables $X_j-m$. 
The result then follows directly from \cite[Theorem 28, p.~286]{petrov:1975}. 
\qed
\end {proof}
The next Lemma is a classical result \cite{bennett:1962}.
\begin {lemma} (Bennett's inequality)
If $W_1,W_2,\ldots$ are independent random variables, 
$\mathbb{E}(W_j) = 0$, $\mathrm{Var}(W_j) = V_W$ and 
$W_j \leq b$ a.s.\ for every $j$, where $b$ and $V_W$ are positive 
numbers, then for any $\epsilon > 0$
\begin {equation} \label{bern}
\P\left(\frac{1}{N}\sum_{j=1}^N W_j > \epsilon\right) \leq  
\exp\left\lbrace-\eta\left[\left(1+\frac{1}
{\lambda}\right)\log\left(1+\lambda\right)
-1\right]\right\rbrace
\end {equation}
with
\begin {equation}
 \eta = \frac{N\epsilon}{b} \quad \quad \text 
 {and} \quad \quad \lambda = \frac{b\epsilon}{V_W}.
\end {equation}
\end {lemma}

By $f(x) = \Theta(g(x))$ as $x\rightarrow \infty$ we mean that for $x$ sufficiently large, there exist constants $c_1$ and $c_2$ such that $c_1 g(x) \leq f(x) \leq c_2 g(x)$.
\begin {lemma} \label{lbern}
If $i = \Theta(N^\gamma)$ where $\gamma < 1$ then, for any $\epsilon > 0$ small 
enough, there is a positive constant $C$ such that

\begin {equation}
\P\left(\frac{S^{(i)}_{N}}{N}-m_i > \epsilon\right) \leq  
\exp\left\lbrace-C\epsilon N^{1-\gamma}\right\rbrace . 
\end {equation}
\end {lemma}
\begin {proof}
This follows directly from Bennett's inequality with 
$W_j = X^{(i)}_{j} - m_i$. 
Then  $V_W = V_i$ and we can take $b=i$ for $i$ large enough 
(since $m_i < 1$ for $i$ large enough). 
If now $i =  \Theta(N^\gamma)$, then
\begin {equation}
 \eta = \epsilon \Theta(N^{1-\gamma}).
\end {equation}
 If $\beta > 3$ then $V_i <\infty$ and $\lambda = \Theta(N^\gamma)$ and 
 the result follows. If $2 <\beta \leq 3$ then 
\begin {equation}  \label{sigmasq}
V_i = \left\{\begin {array}{ll} 
\Theta(i^{3-\beta}) & \text{if $\beta < 3$}, \\
\Theta(\log(i)) & \text{if $\beta = 3$}
\end {array}\right.
\end {equation}
so $\lambda \longrightarrow \infty$ as $N\longrightarrow\infty$ which 
completes the proof. \qed
\end {proof}

In the following we will repeatedly 
use Lagrange's inversion formula, see e.g. \cite[p.~167]{wilf:2006}.
We denote the coefficient of $z^n$
in a formal power series $p(z)$ by $[z^n]\left\lbrace p(z)
\right\rbrace$.

\begin {lemma} (Lagrange's inversion formula)
If $h(z)$ is a formal power series in $z$ and $\L_i$ satisfies 
(\ref{eqtreelower}) then 

\begin {equation} \label {lagrangeinv}
[\zeta^N]\left\lbrace h(\L_i(\zeta))\right\rbrace  = 
\frac{1}{N}[z^{N-1}]\left\lbrace h'(z) \ell_i(z)^N\right\rbrace. 
\end {equation}
\end {lemma}
Using the above Lemma for the function $h(z) = z^j$ we get
\begin {equation}
 [\zeta^N]\left\lbrace \L_i(\zeta)^j\right\rbrace  = 
 \frac{j}{N}[z^{N-j}]\left\lbrace \ell_i(z)^N\right\rbrace. 
\end {equation}

The following simple result will be useful.  We omit the proof.
\begin {lemma} \label{split}
If $X\geq 0$ and $Y$ are random variables, then for any $\epsilon > 0$
\begin {eqnarray}
\P\left(\left|X+Y\right| \leq \epsilon\right)  
\geq \P\left(X \leq \epsilon/2\right) \P\left(\left|Y\right| \leq \epsilon/2\right) 
\end {eqnarray}
and
\begin {eqnarray}
 \P\left(\left|X+Y\right| > \epsilon\right) 
 \leq \P\left(\left|Y\right| > \epsilon/2\right)+\P\left(X > \epsilon/2\right).
\end {eqnarray}
\end {lemma}

\subsection{Calculation of $Z_{1,N}$} \label{ss:Z1N}

Using the Lemmas in the previous subsection we are ready to study 
the asymptotic behaviour of $Z_{1,N}$. It is is easy to see that 
\begin {eqnarray} \label{generalz}
Z_{1,N} &=&  \sum_{i=0}^{N-1}w_{i+1} [\zeta^{N}]\left\lbrace 
\Y_{i}(\zeta)\L_{i}(\zeta)^{i}\right\rbrace
\end {eqnarray}
as is illustrated in Fig.\ \ref{f:z1n}.
\begin{figure} [!b]
\centerline{\scalebox{1}{\includegraphics{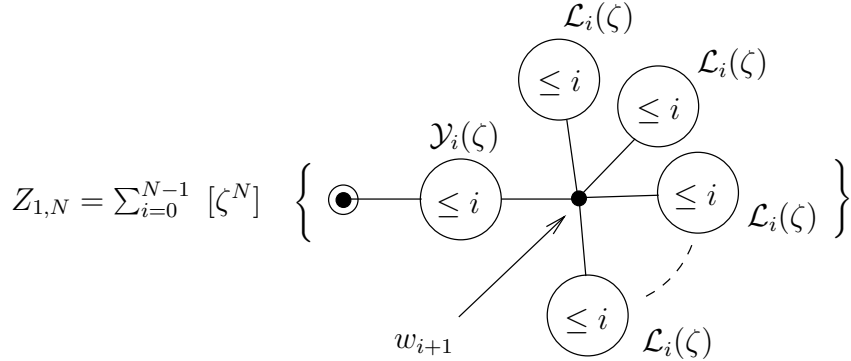}}}
\caption{An illustration of Equation (\ref{generalz}). The 
balloons which include the ``$\leq i$'' are trees which have 
vertices of degree at most $i$. There is thus precisely one 
vertex of maximum degree $i+1$.	} \label{f:z1n}
\end{figure}
Combining (\ref{eqtreelower}) and (\ref{twopoint}) 
one can use the Lagrange inversion formula (\ref{lagrangeinv}) 
for the function
\begin {equation}
h_{ij}(z) = \frac{z^{j+1}}{\ell_i(z)-z\ell'_{i}(z)}
\end {equation}
to get
\begin {eqnarray} \nonumber
[\zeta^{N}]\left\lbrace \Y_{i}(\zeta)\L_{i}(\zeta)^{j}\right\rbrace = \frac{1}{N}  [z^{N-j-1}]\Bigg\{ \Bigg(\frac{j+1}{\ell_i(z)-z\ell_i'(z)} + \frac{z^2\ell_i''(z)}{(\ell_i(z)-z\ell_i'(z))^2}\Bigg)\ell_i(z)^N\Bigg\}. \\ \label{lhij}
\end {eqnarray}
Note that the left hand side in the above equation 
is increasing in $i$ and therefore 
the right hand side also.  In the following we will use this fact repeatedly.  
Next we define the functions
\begin {equation}
f_{i,1}(z) =\frac{\ell_i(1)- \ell_i'(1)}{\ell_i(z)-z\ell_i'(z)}
\end {equation}
and\begin {equation}
f_{i,2}(z) =  \\ \frac{z^2\ell_i''(z)}{(\ell_i(z)-z\ell_i'(z))^2} 
\frac{(\ell_i(1)- \ell_i'(1))^2}{ \ell_i''(1)}.
\end {equation}
It is easy to check that all derivatives of these functions are 
positive for $0 \leq z \leq 1$ and $f_{i,1}(1)=f_{i,2}(1)=1$. 
We let $X^{(i,1)}$ and $X^{(i,2)}$  be random 
variables having $f_{i,1}$ and $f_{i,2}$, respectively, as 
probability generating functions. We will need the following Lemma.

\begin {lemma} \label {gc}
If $i=\Theta(N)$ as $N \longrightarrow \infty$, then for any $\epsilon > 0$ 
\begin{enumerate}
  \item[(i)] $\P\left(X^{(i,1)} \geq \epsilon N\right) \leq C_1 N^{2-\beta}$,          
 
  \item[(ii)] $\ell_N''(1)\P\left(X^{(i,2)} \geq \epsilon N\right) 
  \leq C_2  \left\{\begin {array}{ll} 
N^{3-\beta} & \text{if $\beta \neq 3$}, \\
\log(N) & \text{if $\beta = 3$}
\end {array}\right.$

 \end{enumerate}
where $C_1$ and $C_2$ are positive numbers which in general 
depend on $\epsilon$ and $\beta$.
\end {lemma}
\begin {proof}
 We use a weighted version of Chebyshev's inequality which states that if $X$ is a random variable, $\phi(x)>0$ for $x>0$ 
is monotonically increasing and $\mathbb{E}(\phi(X))$ exists, then
\begin {equation}\label{4.29}
 \P\left(|X|\geq t\right) \leq \frac{\mathbb{E}(\phi(X))}{\phi(t)}.
\end {equation}
We first consider case $(i)$. 
Choose $\phi(x) = x^{\lfloor \beta \rfloor}$ where 
$\lfloor \cdot \rfloor$ denotes the floor function.  It is clear that 
$f^{(n)}_{i,1}(1) < \infty$ for all $n$ and therefore 
$\mathbb{E}(\phi(X^{(i,1)})) 
< \infty$. One can check that as $i \longrightarrow \infty$
\begin {equation} 
 \mathbb{E}(\phi(X^{(i,1)})) = \Theta(\ell_i^{(\lfloor \beta \rfloor +1)}(1)) 
 = \Theta(i^{-\beta + \lfloor \beta \rfloor + 2}).\label{4.30}
\end {equation}
If $i = \Theta(N)$ as $N\longrightarrow \infty$ then by 
(\ref{4.29}) and (\ref{4.30}) there exists a positive constant $C$ such that
\begin {equation}
\P\left(X^{(i,1)} \geq \epsilon N\right)  \leq 
C \dfrac{N^{2-\beta}}{\epsilon^{\lfloor\beta\rfloor}}.
\end {equation}
In order to prove $(ii)$ we first 
consider the case when $2 < \beta \leq 3$. Then 
\begin {equation}  \label{sder}
\ell_N''(1) = \left\{\begin {array}{ll} 
\Theta(N^{3-\beta}) & \text{if $\beta \neq 3$}, \\
\Theta(\log(N)) & \text{if $\beta = 3$}
\end {array}\right.
\end {equation}
as $N\longrightarrow\infty$ which implies the desired result. If
$\beta > 3$, then $\ell_N''(1)$  
has a finite limit as $N\longrightarrow\infty$ 
and the proof proceeds as in case $(i)$. \qed
\end {proof}
We are now ready to prove the main result of this subsection.
\begin {lemma} \label{zsamleitid}
\begin {equation} \label{gzr}
 Z_{1,N} = (1-m)^{-\beta}N^{-\beta} \zeta_0^{1-N}\left(1+o(1)\right).
\end {equation}
\end {lemma}

\begin {proof} In this proof we let $C,C_1,C_2, \ldots$ denote
positive constants 
independent of $N$ whose values may differ between equations. Define
\begin {eqnarray} \nonumber 
G_N(a,b) &=&  g(1)^{1-N} N^{\beta-1}
\sum_{a\leq n \leq b}w_{N-n}[z^{n}] \Bigg\{\ell_{N-n-1}(z)^N 
\\ \nonumber
&&\times~\Bigg( \frac{N-n}{\ell_{N-n-1}(z)-z\ell_{N-n-1}'(z)} + 
\frac{z^2\ell_{N-n-1}''(z)}{(\ell_{N-n-1}(z)-z\ell_{N-n-1}'(z))^2}
\Bigg)\Bigg\}. \\ \label{theG}
\end {eqnarray}
It follows from (\ref{generalz}),  (\ref{lhij}) and (\ref{theG}) that 
\begin {equation}
N^\beta\zeta_0^{N-1}Z_{1,N} = G_N(0,N-1).\label{theG2}
\end {equation}

The strategy of the proof is to split the sum over $n$ on the right
hand side of (\ref{theG2})
into four different parts.  We will see that it is only the region
around $n\approx mN$ which gives a nonvanishing contribution as
$N\to\infty$. 
Choose an $\epsilon > 0$ small enough and a 
$\gamma$ such that $2/\beta<\gamma<1$. 
Then we can write
\begin {eqnarray} \nonumber
N^\beta \zeta_0^{N-1}Z_{1,N} &=& G_N(0,\lfloor (m - \epsilon)N \rfloor) + 
G_N(\lfloor (m - \epsilon)N \rfloor+1,\lfloor (m + \epsilon)N \rfloor) \\  
\nonumber &+& G_N(\lfloor (m + \epsilon)N \rfloor+1,\lfloor N-N^{\gamma} 
\rfloor) +  
G_N(\lfloor N-N^{\gamma} \rfloor+1,N-1). \\ \label{Gsplit}
\end {eqnarray} 
We will show that as $N\longrightarrow\infty$ the second 
term on the right hand side of (\ref{Gsplit})
has a positive limit but the other terms converge to zero. 
To make the notation more compact we define
\begin {equation}
\Np = N-\lfloor(m + \epsilon)N\rfloor-1 \quad \quad \text{and}\quad \quad    
\Nm = N-\lfloor(m - \epsilon)N\rfloor.
\end {equation}

The first term on the right hand side in (\ref{Gsplit}) 
can be estimated from above as follows:
\begin {eqnarray} \nonumber
&& G_N(0,\lfloor (m - \epsilon)N \rfloor)  \\ \nonumber
 &&  \quad \quad \leq \left(\frac{N}{\Nm}\right)^{\beta-1} \sum_{n= 0}^{\lfloor (m - \epsilon)N \rfloor} [z^n]\left\lbrace f(z)^{N}\left(C_1f_{N,1}(z)+\frac{ C_2\ell_N''(1)}{N-n}f_{N,2}(z)\right)\right\rbrace \\ \nonumber
&& \quad \quad \leq  C_3 \P\left(\left|\frac{S_{N}+X^{(N,1)}}{N}-m\right| > \epsilon\right) + \frac{C_4 \ell_N''(1)}{N} \P\left(\left|\frac{S_{N}+X^{(N,2)}}{N}-m\right| > \epsilon\right). \\ \label{lastline}
\end {eqnarray}
where $S_N$ is defined by (\ref{rvdef}).
By Lemma \ref{split} we have for $i=1,2$

\begin {equation} \label{al}
\P\left(\left|\frac{S_{N}+X^{(N,i)}}{N} - m\right| > 
\epsilon \right) \leq \P\left(\left|\frac{S_{N}}{N} - m\right| > 
\epsilon/2 \right) + 
\P\left( X^{(N,i)} > N\epsilon/2 \right).
\end {equation}
This, combined with (\ref{sder}) and  Lemmas \ref{lemmabk} 
and \ref{gc}, shows that the 
two terms on the right hand side of (\ref{lastline}) go to 
zero as $N \longrightarrow\infty$.
\medskip
\medskip

We estimate the third term on the right hand side of 
(\ref{Gsplit}) from above as follows:

\begin {eqnarray} \nonumber
&& G_N(\lfloor (m + \epsilon)N \rfloor+1,\lfloor N-N^{\gamma} \rfloor) \leq \left(\frac{N}{N-\lfloor N-N^\gamma\rfloor}\right)^{\beta-1} \\ \nonumber
&&  \quad \quad \times ~\sum_{n = \lfloor(m + \epsilon)N\rfloor+1}^{\lfloor N-N^{\gamma}\rfloor} [z^n]\left\lbrace f(z)^{N}\left(C_1f_{N,1}(z)+\frac{ C_2\ell_N''(1)}{(N-\lfloor N-N^\gamma\rfloor)}f_{N,2}(z)\right)\right\rbrace \\ \nonumber
&& \leq C_3 N^{(1-\gamma)(\beta-1)} \P\left(\left|\frac{S_{N}+X^{(N,1)}}{N}-m\right| > \epsilon\right) \\ \label{abovethird}
&&  \quad \quad  + ~ C_4N^{(1-\gamma)(\beta-1)-\gamma}\ell''_N(1)\P\left(\left|\frac{S_{N}+X^{(N,2)}}{N}-m\right| > \epsilon\right). \label{rhs444}
\end {eqnarray}
Since $\gamma > 2/\beta$ it holds that $(1-\gamma)(\beta-1) < \beta -2$ 
and $(1-\gamma)(\beta-1) - \gamma < \beta - 3$. 
Then by (\ref{sder}), (\ref{al}) 
and Lemmas \ref{lemmabk} and \ref{gc} we see that last two terms on
the right hand side of (\ref{rhs444}) converge to 
zero as $N \longrightarrow\infty$.

\medskip
\medskip 

To estimate the fourth term of (\ref{Gsplit}) from 
above we first note that 
\begin {equation}
 [\zeta^N]\left\{\Y_i(\zeta)\right\} = 
 [\zeta^N]\left\{\frac{\partial}{\partial w_1}
 \L_i(\zeta)\right\} \leq \frac{N}{w_1}[\zeta^N]\left\{  \L_i(\zeta)\right\}
\end {equation}
and thus
\begin {equation} \label{mat}
G_N(a,b) \leq w_1^{-1}g(1)^{1-N}N^{\beta} 
\sum_{a\leq n \leq b}w_{N-n}(N-n)[z^{n}] 
\left\{\ell_{N-n-1}(z)^N \right\}
\end {equation}
for any $a,b$.
Using (\ref{mat}) for $N$ large enough and $\epsilon$ small enough
(but independent of $N$) we get
\begin {eqnarray} \nonumber
G_N(\lfloor N-N^{\gamma} \rfloor+1,N-1) &\leq& C_1 N^{\beta} \sum_{n= 
\lfloor N-N^{\gamma}\rfloor+1}^{N-1} [z^n]\left\lbrace f_{N-\lfloor N-N^{\gamma} 
\rfloor}(z)^{N}\right\rbrace \\ \nonumber
&=& C_1 N^{\beta} \P\left( \lfloor N-N^{\gamma}
\rfloor +1 \leq S^{(N-\lfloor N-N^{\gamma}
\rfloor)}_{N} \leq N-1 \right)  \\
\nonumber
&\leq& C_1 N^{\beta} \P\left(\frac{S^{(N-\lfloor N-N^{\gamma} \rfloor)}_{N}}{N}-m_{N-\lfloor 
N-N^{\gamma} \rfloor} \geq \epsilon\right)  \\ 
&\leq& C_1 N^{\beta} \exp\left(-C_2\epsilon N^{1-\gamma}\right) \label{abovefourth}
\end {eqnarray}
where in the last step we used Lemma \ref{lbern}. The last expression converges to 
zero as $N\to \infty$ since $\gamma < 1$. 

Finally, we show that the second term in (\ref{Gsplit}) 
has a nonzero contribution as $N \to \infty$. By (\ref{bw}) 
we see that for $n$ large enough we have 
\begin{equation}
(1-\epsilon)n^{-\beta} \leq w_n \leq (1+\epsilon)n^{-\beta}.
\end{equation}
We then get the upper bound
\begin {eqnarray} \nonumber
&&G_N(\lfloor (m - \epsilon)N \rfloor+1,\lfloor (m + \epsilon)N \rfloor) \leq (1+\epsilon)g(1) \left(\frac{N}{\Np}\right)^{\beta-1} \\ \nonumber
&& \quad \quad  \times ~\Bigg(\frac{1}{\ell_N(1)-\ell_N'(1)}\sum_{n=\lfloor (m - \epsilon)N\rfloor + 1}^{\lfloor(m + \epsilon)N\rfloor}  [z^n]\left\lbrace f_{N,1}(z)f(z)^{N}\right\rbrace \\ \nonumber
&& \quad \quad  + ~\frac{ \ell_N''(1)}{(\ell_N(1)- \ell_N'(1))^2 \Np}\sum_{n=\lfloor (m - \epsilon)N\rfloor + 1}^{\lfloor(m + \epsilon)N\rfloor}  [z^n]\left\lbrace f_{N,2}(z)f(z)^{N}\right\rbrace\Bigg)\\ \nonumber
&&  \leq (1+\epsilon)g(1) \left(\frac{N}{\Np}\right)^{\beta-1}\left(\frac{1}{\ell_N(1)-\ell_N'(1)}+\frac{ \ell_N''(1)}{(\ell_N(1)- \ell_N'(1))^2 \Np}\right)\\
&&  \longrightarrow \frac{(1+\epsilon)(1-(m+\epsilon))^{1-\beta}}{1-m} \label{abovesecond}
\end {eqnarray}
as $N \longrightarrow \infty$ by (\ref{sder}). In a similar way we get the lower bound
\begin {eqnarray} \nonumber
&&G_N(\lfloor (m - \epsilon)N \rfloor+1,\lfloor 
(m + \epsilon)N \rfloor) \geq (1-\epsilon) g(1) 
\left(\frac{N}{\Nm}\right)^{\beta-1} 
\left(\frac{\ell_{\Np}(1)}{g(1)}\right)^{N}\\ \nonumber
&& \times ~\Bigg(\frac{1}{\ell_{\Np}(1)-\ell_{\Np}'(1)}
\sum_{n=\lfloor (m - \epsilon)N\rfloor + 1}^{\lfloor(m + \epsilon)N\rfloor} 
[z^n]\left\lbrace f_{{\Np},1}(z)f_{\Np}(z)^{N}\right\rbrace \\ \nonumber
&& + ~\frac{ \ell_{\Np}''(1)}{(\ell_{\Np}(1)- 
\ell_{\Np}'(1))^2 \Nm}\sum_{n=\lfloor (m - \epsilon)N\rfloor + 1}^{\lfloor(m 
+ \epsilon)N\rfloor}   [z^n]\left\lbrace f_{{\Np},2}(z)f_{\Np}(z)^{N}
\right\rbrace\Bigg). \\ \label{belowsecond1}
\end {eqnarray}
By (\ref{sder}) the second term in the parenthesis above converges to zero as 
$N\longrightarrow\infty$. Looking at the first term we find that
\begin {equation} \label{belowsecond2}
 \frac{(1-\epsilon)g(1)}{\ell_{\Np}(1)-\ell_{\Np}'(1)} 
 \left(\frac{N}{\Nm}\right)^{\beta-1} \longrightarrow 
 \frac{(1-\epsilon)(1-(m-\epsilon))^{1-\beta}}{1-m}
\end {equation}
as $N\longrightarrow\infty$ and
\begin{equation} \label{belowsecond3}
\left(\frac{\ell_{\Np}(1)}{g(1)}\right)^{N} = 
\left(1-\frac{1}{g(1)}\sum_{n=\Np}^\infty w_{n+1} \right)^{N} = 
\left(1 + \Theta(N^{-\beta+1})\right)^{N} \longrightarrow 1 
\end{equation}
as $N\longrightarrow\infty$ since $\beta > 2$. Finally,
we have for $N$ large enough

\begin {eqnarray} \nonumber
&&\sum_{n=\lfloor (m - \epsilon)N\rfloor + 1}^{\lfloor(m + \epsilon)N\rfloor}   [z^n]\left\lbrace f_{{\Np},1}(z)f_{\Np}(z)^{N}\right\rbrace \\ \nonumber
&=& \P\left(\left|\frac{S^{(\Np)}_{N}+X^{(\Np,1)}}{N} - m\right| \leq \epsilon\right) \\ \nonumber
&\geq& \P\left(\left|\frac{S^{(\Np)}_{N}+X^{(\Np,1)}}{N} - m_{\Np}\right| \leq \epsilon/2\right) \\ \nonumber
&\geq& \P\left(\left|\frac{S^{(\Np)}_{N}}{N} - m_{\Np}\right| \leq \epsilon/4\right)\P\left(X^{(\Np,1)} \leq N \epsilon/4\right)\\ \label{belowsecond4}
&\geq& \left(1-\frac{V_{\Np}}{N\left(\epsilon/4\right)^2}\right)
\left(1-C N^{2-\beta}\right)
\end {eqnarray}
where in the second last step we used Lemma \ref{split} 
and in the last step we used Chebyshev's inequality and 
Lemma \ref{gc}. As $N \longrightarrow \infty$ it is clear 
from (\ref{sigmasq}) that $V_{\Np}/N \longrightarrow 0$ 
and therefore the last expression converges to $1$.

From the above estimates (\ref{lastline}), (\ref{abovethird}) 
and (\ref{abovefourth}--\ref{belowsecond4}) we find that 

\begin {eqnarray}
\frac{(1-\epsilon)\left(1-\left(m-\epsilon \right)\right)^{1-\beta}}{1-m} 
&\leq& \liminf_{N\rightarrow\infty} N^\beta \zeta_0^{N-1}Z_{1,N}
\nonumber\\
&\leq& \limsup_{N\rightarrow\infty} 
N^\beta \zeta_0^{N-1}Z_{1,N} 
\leq \frac{(1+\epsilon)\left(1-\left(m+\epsilon
\right)\right)^{1-\beta}}{1-m}.\nonumber \\
\label{infsup}
\end {eqnarray}
Since this holds for all $\epsilon > 0$, small enough, we have

\begin {equation}
\lim_{N\rightarrow\infty}  N^\beta \zeta_0^{N-1}Z_{1,N} = 
\left(1-m\right)^{-\beta}
\end {equation}
which completes the proof. \qed
\end {proof}

\subsection {Estimate on $E_N$} \label{ss:EN}
We now estimate $E_N$, the remaining contribution to $Z_N$. 
Note that $\L_{i+1}(\zeta)-\L_{i}(\zeta)$ is the grand canonical 
partition function for trees which have at least one vertex of 
degree $i+1$ and no vertex of degree greater than $i+1$. 
Consider a tree which has $\geq 2$ vertices of maximal degree $i+1$. 
Denote the two maximal degree vertices closest to the root and 
second closest to the root by $s_1$ and $s_2$, respectively. These vertices 
are not necessarily unique and can be at the same distance from the
root, but for the following purpose we can choose any two we like. 
Denote the path from the root to $s_2$ by $(r,s_2)$.   Then either
$s_1$ is on $(r,s_2)$ or it is not so we can write

\begin{figure} [!b]
\centerline{\scalebox{0.8}{\includegraphics{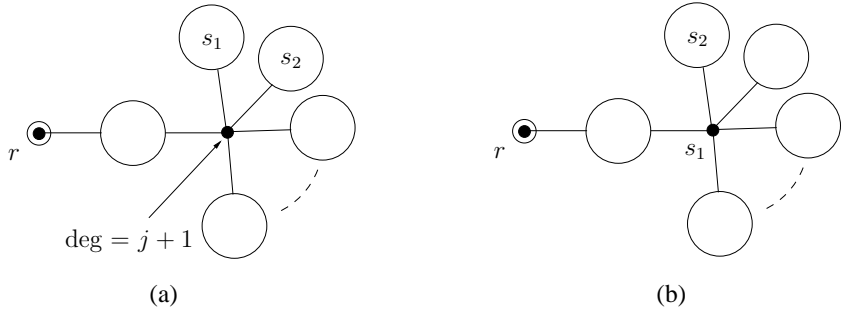}}}
\caption{{\bf a)} The case when $s_1 \notin (r,s_2)$. At least 
two balloons attached to the vertex of degree $j+1$ 
(excluding the rooted one) indicated in the figure have to have 
at least one vertex of degree $i+1$, namely $s_1$ and $s_2$. 
{\bf b)} The case when $s_1 \in (r,s_2)$. At least one balloon 
attached to the vertex $s_1$ (excluding the rooted one) 
has to have at least one vertex of degree $i+1$ , namely $s_2$. } \label{e1}
\end{figure}

\begin {eqnarray}\nonumber
 E_N &=& \sum_{i=0}^{N-1}\Bigg(\sum_{j=0}^{i-1}w_{j+1}  [\zeta^{N}]\left\lbrace \Y_{i}(\zeta)\sum_{n=2}^j \binom{j}{n}\underbrace{\left(\L_{i+1}(\zeta)-\L_{i}(\zeta)\right)^n}_{\text{$s_1$ and $s_2$ in here}} \L_i(\zeta)^{j-n}\right\rbrace \\ \nonumber
&& \quad \quad \quad + ~\underbrace{w_{i+1}}_{s_1}  [\zeta^{N}]\left\lbrace \Y_{i}(\zeta)\sum_{n=1}^i \binom{i}{n}\underbrace{\left(\L_{i+1}(\zeta)-\L_{i}(\zeta)\right)^n}_{\text{$s_2$ in here}} \L_i(\zeta)^{i-n}\right\rbrace \Bigg). \\ \label{errorterm}
\end {eqnarray}
The outermost sum is over all possible maximal degrees. The first 
term in the brackets takes care of the case when $s_1 \notin (r,s_2)$. 
Then $j+1$ is the degree of the vertex where $(r,s_1)$ and $(r,s_2)$ 
start to differ. At least two of the subtrees attached to this vertex 
(excluding the rooted one) have to have at least one vertex of degree 
$i+1$, see Figure \ref{e1} (a). The second term in the brackets takes 
care of the case when $s_1 \in (r,s_2)$. At least one of the subtrees 
attached to $s_1$ (excluding the rooted one) has to have at least 
one vertex of degree $i+1$, see Figure \ref{e1} (b).

\begin {lemma} \label{ie1}
For any $i$ and $N$ we have
\begin {equation}
[\zeta^N] \left\lbrace \L_{i+1}(\zeta)-\L_{i}(\zeta) 
\right\rbrace \leq \frac{w_{i+1}N}{i}[\zeta^N] \left\lbrace 
\zeta \L_{i+1}(\zeta)^i \right\rbrace. 
\end {equation}
\end {lemma}
\begin {proof}
Use the Lagrange inversion theorem to obtain
\begin {eqnarray}
 [\zeta^N] \left\lbrace \L_{i+1}(\zeta)-\L_{i}(\zeta) \right\rbrace &=& 
 \frac{1}{N} [z^{N-1}] \left\lbrace \ell_{i+1}(z)^N - 
 \ell_{i}(z)^N\right\rbrace \nonumber \\
&=& \frac{1}{N}[z^{N-1}]\left\lbrace \left(\ell_{i+1}(z)-
\ell_{i}(z)\right)\sum_{N_1+N_2 = N-1} \ell_{i+1}(z)^{N_1}\ell_{i}(z)^{N_2} 
\right\rbrace \nonumber \\
&\leq& w_{i+1}[z^{N-i-1}]\left\lbrace \ell_{i+1}(z)^{N-1} \right\rbrace
\label{bigbrace}
\end {eqnarray}
Now use the Lagrange inversion theorem on the right hand side of 
(\ref{bigbrace}) to obtain the result. \qed
\end {proof}
\begin {lemma} \label{uone}
For any $N$ we have
\begin {equation}
  E_N \leq 2 N^2\sum_{i=0}^{N-1}w_{i+1}^2 [\zeta^{N-1}] 
  \left\lbrace \Y_{i+1}(\zeta) \L_{i+1}(\zeta)^{2i-1}\right\rbrace.
  \label{ENest}
 \end {equation}
\end {lemma}
\begin {proof}
 First note that
\begin {eqnarray}
 &&\sum_{n=2}^j \binom{j}{n}\left(\L_{i+1}(\zeta)-\L_{i}(\zeta)\right)^n 
 \L_i(\zeta)^{j-n}\nonumber\\ 
 &=& \L_{i+1}(\zeta)^j-\L_{i}(\zeta)^j - j(\L_{i+1}(\zeta)-\L_i(\zeta))
 \L_i(\zeta)^{j-1} \nonumber\\
&=& (\L_{i+1}(\zeta)-\L_i(\zeta))\left(\sum_{j_1+j_2=j-1}
\L_{i+1}(\zeta)^{j_1}\L_i(\zeta)^{j_2}-j\L_i(\zeta)^{j-1}\right)
\nonumber\\
&\leq& j(\L_{i+1}(\zeta)-\L_i(\zeta))\left(\L_{i+1}(\zeta)^{j-1}-
\L_i(\zeta)^{j-1}\right) \nonumber\\
&=& j (\L_{i+1}(\zeta)-\L_i(\zeta))^2 \sum_{j_1+j_2=j-2}
\L_{i+1}(\zeta)^{j_1}\L_i(\zeta)^{j_2} \nonumber\\
&\leq& j(j-1)(\L_{i+1}(\zeta)-\L_i(\zeta))^2
\L_{i+1}(\zeta)^{j-2}.\label{longest} 
\end {eqnarray}
It is also clear that the above inequality holds inside 
$[\zeta^N]\left\lbrace \cdot \right\rbrace$ brackets. 
Therefore the sum over $j$ in (\ref{errorterm}) is estimated 
from above by 
\begin {eqnarray}
&& \sum_{j=0}^{i-1}w_{j+1}  [\zeta^{N}]\left\lbrace \Y_{i}(\zeta)\sum_{n=2}^j 
\binom{j}{n}\left(\L_{i+1}(\zeta)-\L_{i}(\zeta)\right)^n  
\L_i(\zeta)^{j-n}\right\rbrace \nonumber\\
&\leq& [\zeta^N] \left\lbrace   \Y_{i}(\zeta) (\L_{i+1}(\zeta)-
\L_i(\zeta))^2 \ell_{i}''(\L_{i+1}(\zeta))\right\rbrace.\label{bibi}
\end {eqnarray}
Now use Lemma \ref{ie1} to get
\begin {eqnarray}
&& [\zeta^N] \left\lbrace   \Y_{i}(\zeta) (\L_{i+1}(\zeta)-\L_i(\zeta))^2 
\ell_{i}''(\L_{i+1}(\zeta))\right\rbrace \nonumber\\
&=& \sum_{N_1+N_2+N_3 = N} [\zeta^{N_1}]\left\lbrace   
\Y_{i}(\zeta)  \ell_{i}''(\L_{i+1}(\zeta))\right\rbrace[\zeta^{N_2}]
\left\lbrace \L_{i+1}(\zeta)-\L_i(\zeta)\right\rbrace[\zeta^{N_3}]
\left\lbrace  \L_{i+1}(\zeta)-\L_i(\zeta)\right\rbrace \nonumber\\
&\leq& \frac{w_{i+1}^2}{i^2}N^2\sum_{N_1+N_2+N_3 = N} [\zeta^{N_1}]
\left\lbrace   \Y_{i}(\zeta)  \ell_{i}''(\L_{i+1}(\zeta))
\right\rbrace[\zeta^{N_2}]\left\lbrace \zeta \L_{i+1}(\zeta)^i
\right\rbrace[\zeta^{N_3}]\left\lbrace  \zeta \L_{i+1}(\zeta)^i
\right\rbrace \nonumber\\
&=& \frac{w_{i+1}^2}{i^2}N^2[\zeta^{N}]\left\lbrace  \zeta^2 \Y_{i}(\zeta)  
\ell_{i}''(\L_{i+1}(\zeta)) \L_{i+1}(\zeta)^{2i}\right\rbrace.
\label{diffest}
\end {eqnarray}
Next observe that
\begin {equation}
\zeta \frac{\ell_{i}''(\L_{i+1}(\zeta))}{i^2} \L_{i+1}(\zeta) \leq \frac{\zeta \ell_{i+1}(\L_{i+1}(\zeta))}{\L_{i+1}(\zeta)} = 1.
\end {equation}
Combining the above results we have the estimate
\begin {eqnarray}
&& \sum_{j=0}^{i-1}w_{j+1}  [\zeta^{N}]\left\lbrace \Y_{i}(\zeta)
\sum_{n=2}^j \binom{j}{n}\left(\L_{i+1}(\zeta)-\L_{i}(\zeta)\right)^n  
\L_i(\zeta)^{j-n}\right\rbrace \nonumber\\
&\leq&  w_{i+1}^2 N^2[\zeta^{N-1}] \left\lbrace \Y_{i+1}(\zeta) 
\L_{i+1}(\zeta)^{2i-1}\right\rbrace \label{finalest} 
\end {eqnarray}
We get precisely the same estimate for the term in the second line in 
(\ref{errorterm}) (the calculations are even simpler) except 
that it is of order $N$ smaller and (\ref{ENest}) follows.  \qed
\end {proof}
The above lemma implies the following result.
\begin {lemma} \label{ehverfur}
\begin {equation}
N^\beta \zeta_0^N E_N \longrightarrow 0 
\quad \quad \text{as} \quad \quad N \longrightarrow \infty. 
\end {equation}
\end {lemma}
\begin {proof}
 By Lemma \ref{uone}
\begin {equation}
 N^\beta \zeta_0^N E_N \leq 2 N^{\beta+2} \zeta_0^N 
 \sum_{i=0}^{N-1}w_{i+1}^2 [\zeta^{N-1}] \left\lbrace \Y_{i+1}(\zeta) 
 \L_{i+1}(\zeta)^{2i-1}\right\rbrace. 
\end {equation}
The sum on the right hand side has the same form as $Z_{1,N}$ with $\beta$ replaced by $2\beta$, cf.~Equation (\ref{generalz}). Equation (\ref{gzr}), which describes the asymptotic behaviour of $Z_{1,N}$, can therefore be applied to show that the right hand side is  $o(N^{2-\beta})$.  Since $\beta > 2$, this converges to zero as $N \longrightarrow \infty$.  \qed
\end {proof}
Combining Lemmas \ref{zsamleitid} and \ref{ehverfur} 
completes the proof of Theorem \ref{th1}.
\subsection{Generalization of $Z_N$} 
For technical reasons, which will be made clear in the next section, 
we need to generalize the sequence $Z_N$ as we now describe. 
If $r$ is the root of a tree we denote its unique nearest neighbour by
$s$.   Define
\begin {equation}
Z_N^{(R)} = \sum_{\tau \in \Gamma_N} w_{\sigma_s+R-1} 
\prod_{i\in V(\tau)\setminus \{r,s\}} w_{\sigma_i}.
\end {equation}
In analogy with (\ref{grandcanonical}) and (\ref{potential}), define the generating functions
\begin {equation}
 \Z(\zeta,R) = \sum_{N=1}^\infty Z_N^{(R)} \zeta^N
\end {equation}
and
\begin {equation}
 g_R(z) = \sum_{n=0}^\infty w_{n+R}z^n.
\end {equation}
Clearly $Z_N = Z_N^{(1)}, \Z(\zeta) = \Z(\zeta,1)$ 
and $g(z) = g_1(z)$.  By the same arguments as for (\ref{eqtree}) 
we find the relation
\begin {equation}
 \Z(\zeta,R) = \zeta g_R(\Z(\zeta)).
\end {equation}
Let $\Z_{0,R} = \Z(\zeta_0,R)$.
The following Lemma is a generalization of Theorem \ref{th1}.
\begin{lemma} \label{l:ZRN}
For the branching weights (\ref{bw}) which satisfy 
(\ref{ngcondition}) it holds that
 \begin {equation} 
 Z_{N}^{(R)} = \left(1-m+\frac{g_{R}'(1)}{g(1)}\right)(1-m)^{-\beta} 
 N^{-\beta} \zeta_0^{1-N}\left(1+o(1)\right).
\end {equation}
\end{lemma}

\begin{proof}
We write 
\begin {equation}
Z_{N}^{(R)} = Z_{1,N}^{(R)} + E_N^{(R)}
\end {equation}
in analogy with (\ref{splitsum}). One can show with the same 
methods as in the previous subsection that 
$\lim_{N\rightarrow\infty} E_N^{(R)}/Z_N = 0$. Therefore 
we focus on the term $Z_{1,N}^{(R)}$, the contribution from 
trees with exactly one vertex of maximal degree. We split this 
term into two parts: one where the maximal degree vertex is the nearest 
neighbour of the root and another when it is not. We can then write
\begin {eqnarray} \nonumber
 Z_{1,N}^{(R)} &=& \sum_{i=0}^{N-1} w_{i+R} [\zeta^N]\left\{\zeta 
 \L_i(\zeta)^i\right\} + \sum_{i=0}^{N-2} w_{i+1} 
 [\zeta^N]\left\{\zeta \ell_{i,R}'(\L_i(\zeta)) 
 \Y_i(\zeta) \L_i(\zeta)^i\right\} \\ \label{zr}
\end {eqnarray}
where we have defined
\begin {equation}
 \ell_{i,R}(z)  = \sum_{n=0}^{i-1}w_{n+R}z^{n}.
\end {equation}
Let
\begin {equation}
h(z) = \frac{z^{i+1}}{\ell_{i}(z)}
\end {equation}
and
\begin {equation}
k(z) = \frac{\ell'_{i,R}(z) z^{i+2}}{\ell_i(z)\left(\ell_i(z)-
z\ell'_i(z)\right)}.
\end {equation}
Using the Lagrange inversion formula for the functions $h$ 
and $k$ we find that
\begin {equation} \label{zn1}
[\zeta^{N}]\left\{\zeta \L_i(\zeta)^i\right\} = 
\frac{1}{N} [z^{N-i-1}]\left\{\left(\frac{i+1}{\ell_i(z)}-
\frac{z\ell_i'(z)}{\ell_i(z)^2}\right)\ell_i(z)^N\right\}
\end {equation}
and
\begin {eqnarray} \nonumber
[\zeta^N]\left\{\zeta \ell_{i,R}'(\L_i(\zeta)) \Y_i(\zeta) 
\L_i(\zeta)^i\right\} &=& \frac{1}{N} [z^{N-i-2}]\Bigg\{\Bigg(\frac{(i+2)
\ell_{i,R}'(z)}{\ell_i(z)(\ell_i(z)-z\ell_i'(z))}\\ \nonumber
&&+~z\frac{d}{dz}\Bigg(\frac{\ell_{i,R}'(z)}
{\ell_i(z)(\ell_i(z)-z\ell_i'(z))}\Bigg)\Bigg)\ell_i(z)^N\Bigg\}. 
\\ \label{zn2}
\end {eqnarray}
We now use exactly the same arguments as in the proof of 
Lemma \ref{zsamleitid} to estimate the asymptotic behaviour 
of (\ref{zr}).  One can show that the contribution from the 
second term in the curly brackets in (\ref{zn1}) and (\ref{zn2}) 
is negligible. Then one can show that for any $\epsilon > 0$

\begin {eqnarray}
\liminf_{N\rightarrow\infty} N^\beta 
\zeta_0^{N-1}Z_{1,N}^{(R)} \geq (1-\epsilon)
\left(1-\left(m-\epsilon \right)\right)^{1-\beta} 
\left(1+\frac{g_{R}'(1)}{g(1)-g'(1)	}\right)
\end {eqnarray}
and
\begin {equation}
\limsup_{N\rightarrow\infty} N^\beta \zeta_0^{N-1}Z_{1,N}^{(R)} 
\leq (1+\epsilon)\left(1-\left(m+\epsilon \right)\right)^{1-\beta} 
\left(1+\frac{g_{R}'(1)}{g(1)-g'(1)}\right).
\end {equation}
Since this holds for all $\epsilon > 0$ the desired result follows. \qed
\end {proof}

\section {The infinite volume limit} \label{s:measure}
In this section we show that the 
measures $\nu_N$ converge as $N\to\infty$ and we characterize the limits
for the three different cases 
discussed in Section \ref{s:themodel}.  If  $m$ 
is the mean offspring probability defined in 
(\ref{meanoffspring}) then the three cases are: generic, 
critical case ($w_1 < w_c$, $m=1$), the nongeneric, critical 
case ($w_1 = w_c$, $m=1$) and the nongeneric, subcritical 
case ($w_1 > w_c$, $m<1$). 

All the results stated for generic trees have already been established
\cite{durhuus:2007} 
but are rederived here in a slightly different way. In the 
generic case, Equation (\ref{eqtree}) can be solved for $\Z(\zeta)$ 
close to the critical point $\zeta_0$ and one can then find 
the asymptotic behaviour of $Z_N$, the coefficients of 
$\Z(\zeta)$, see \cite[Theorem 3.1]{meir:1978}. In 
the non--generic critical case, the function $\Z(\zeta)$ has 
the same critical behaviour as in the generic case as long 
as $g''(1)<\infty$, see \cite[Lemma A.2]{janson:2006}.  
By the same arguments as in~\cite{flajolet:2009,janson:2006} 
one finds the following result for $Z_N^{(R)}$.
\begin {lemma} \label{l:zgeneric}  
Under the stated assumption on the branching weights 
(\ref{statedassumption}) and assuming that $m=1$ and 
$g''(\Z_0)<\infty$ it holds that
\begin {equation}
Z_N^{(R)} = \sqrt{\frac{g(\Z_0)}{2\pi g''(\Z_0)}}\zeta_0 
g_R'(\Z_0)N^{-3/2} \zeta_0^{-N}\left(1+o(1)\right).
\end {equation}
\end {lemma}

An analogous result for the asymptotic behaviour of $Z_N$, for 
a special choice of branching weights corresponding to nongeneric 
critical trees with $g''(1)=\infty$, is stated in~\cite[VI.18 and VI.19, page 407]{flajolet:2009}. A generalization to $Z_N^{(R)}$ 
is straightforward and is stated in the following Lemma.

\begin {lemma} \label{ngcrit}
For the nongeneric, critical branching weights defined 
by (\ref{bw}), with $2 < \beta < 3$ and $w_1=w_c$ we have
\begin {equation}
Z_N^{(R)} = C \zeta_0 g_R'(1)  N^{-\frac{\beta}{\beta-1}}
\zeta_0^{-N}\left(1+ o(1)\right) 
\end {equation}
where $C>0$ is a constant.
\end {lemma}

We now prove that the measures $\nu_N$ converge as $N\rightarrow \infty$
provided that $Z^{(R)}_N$ has the right asymptotic behaviour.  We call a self avoiding, infinite, linear path starting at the root a spine.
\begin {theorem} \label{t:conv}
If
\begin {equation} \label{asymZ}
 Z^{(R)}_N = C \left(1-m+\zeta_0 g_R'(Z_0)\right) 
 N^{-\delta}\zeta_0^{-N}(1+o(1))
\end {equation}
where $C$ is a positive constant and $\delta > 1$, then the 
measures $\nu_N$ converge weakly, as $N \longrightarrow \infty$, 
to a probability measure $\nu$ which has the following properties:
\begin {itemize}
 \item If $m=1$, $\nu$ is concentrated on the set of trees with 
 exactly one spine having finite, independent, critical 
 GW outgrowths defined by the offspring 
 probabilities in (\ref{GWprob}). The numbers $i$ and $j$ 
 of left and right outgrowths from a vertex on the spine are 
 independently distributed by
\begin {equation} \label{sss}
\phi(i,j) = \frac{1}{m}\zeta_0 w_{i+j+2} Z_0^{i+j}.
\end {equation}
\item If $m<1$, $\nu$ is concentrated on the set of trees 
with exactly one vertex of infinite degree which we denote by $t$. 
The length $\ell$ of the path $(r,t)$ is distributed by
\begin {equation} 
\psi(\ell) = (1-m)m^{\ell-1}.
\end {equation}
The outgrowths from the path $(r,t)$ are finite, independent, 
subcritical GW trees defined by the offspring 
probabilities in (\ref{GWprob}). The numbers $i$ and $j$ of left 
and right outgrowths from a vertex $v\in(r,t), v\neq t$ are 
independently distributed by (\ref{sss}).
\end {itemize}
\end {theorem}
\begin {proof}
First we prove existence of $\nu$. Since the metric space $(\Gamma,d)$ 
has the properties stated in Propositions 
(\ref{p:openclosed} and \ref{p:compact}) it is enough,
as explained in \cite{billingsley:1968,durhuus:2003}, to show that 
for any $k \in \mathbb{N}$ and $\tau' \in \Gamma'$ the probabilities
\begin {equation} \label{pc}
 \nu_N\left(\mathcal{B}_{\frac{1}{k}}\left(\tau'\right)\right)
\end {equation}
converge as $N\longrightarrow\infty$. Since $\Gamma$ is 
compact, tightness is automatically fulfilled. The ball 
in (\ref{pc}) can be expressed as
\begin {equation}\label{ball}
 \mathcal{B}_{\frac{1}{k}}\left(\tau'\right) = 
 \left\{\tau\in\Gamma ~|~ L_R(\tau) = \tau_0 \right\}
\end {equation}
where  $R=k+1$ and $\tau_0 = L_R(\tau')$. Denote the number of 
vertices in $\tau_0$ of degree $R$ by $S$ and the number of 
vertices in $\tau_0$ at distance $R$ from the root by $T$. It is 
clear that $S+T\geq 0$.
\begin {figure} [!t]
\centerline{{\scalebox{0.8}{\includegraphics{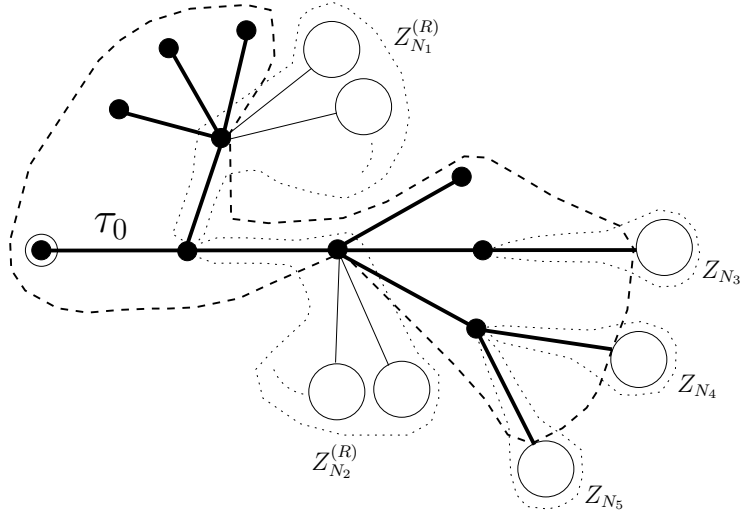}}}}
\caption{An example of the set (\ref{ball}) where $R=4$, 
$S=2$ and $T=3$. When conditioning on trees of size $N$ one 
attaches the weights $Z^{(R)}_{N_i}$, $i=1\ldots S$ 
and $Z_{N_j}$, $j=S+1\ldots S+T$ as indicated in the figure.}\label{f:conv}
\end{figure}
We can now write
\begin {eqnarray} \nonumber
  &&\nu_N\left(\left\{\tau\in\Gamma ~|~ L_R(\tau) = 
  \tau_0 \right\}\right) = \\ \label{probball} && \quad \quad \quad 
  Z_N^{-1} W_0 \sum_{N_1+\ldots+N_{S+T} = N-|\tau_0| + T + S} ~\prod_{i=1}^S 
  Z_{N_i}^{(R)} \prod_{j=S+1}^{S+T}Z_{N_j}
\end {eqnarray}
where
\begin {equation} \label{W0}
W_0 = \prod_{\substack{v \in V(\tau_0) \setminus \{r\} \\ 
\sigma_v,|(r,v)| \neq R}} w_{\sigma_v}
\end {equation}
is the weight of the tree $\tau_0$ (except the contribution
from the vertices 
which are explicitly excluded), and $|(r,v)|$ denotes the 
length of the path $(r,v)$, see Fig.~\ref{f:conv}. For one of the 
indices $k$ in each term of the above sum it holds that 
$N_k \geq \frac{N-|\tau_0| + S + T}{S+T}$. Consider the contribution 
from terms for which $N_n > A$ for some other index $n\neq k$ 
and $A>0$. The indices $n$ and $k$ belong to one of 
the sets $\{1,\ldots,S\}$ and $\{S+1,\ldots, S+T\}$, a total of four 
possibilities. First assume that $S \geq 2$ and $n,k \in \{1,\ldots,S\}$. 
Using (\ref{asymZ}), this contribution can be estimated from above by
\begin {eqnarray}
 &&C_1 \zeta_0^{N} Z_N S^2\sum_{\substack{N_1+\ldots+N_{S+T} = 
 N-|\tau_0| + T + S\\N_1 \geq \frac{N-|\tau_0| + S + T}{S+T}, ~N_2 > A}} 
 Z_{N_1}^{(R)}\zeta_0^{N_1}~\prod_{i=2}^S Z_{N_i}^{(R)}\zeta_0^{N_i} 
 \prod_{j=S+1}^{S+T}Z_{N_j}\zeta_0^{N_j} \nonumber \\
&& \leq C_2  \left(\frac{(S+T)N}{N-|\tau_0| +T +S}\right)^{\delta} 
\sum_{\substack{N_3,\ldots,N_{S+T} \geq 1\\N_2 > A}} 
~\prod_{i=2}^S Z_{N_i}^{(R)}\zeta_0^{N_i} \prod_{j=S+1}^{S+T}Z_{N_j}
\zeta_0^{N_j} \nonumber\\ 
&& \leq C_3 \Z_{0,R}^{S-2} \Z_0^T \sum_{N_2>A} N_2^{-\delta} 
\leq C_4 A^{1-\delta}
\end {eqnarray}
where $C_1$, $C_2$, $C_3$ and $C_4$ are positive numbers 
independent of $N$ and $A$. Exactly the same upper bound is obtained, 
up to a multiplicative constant, for the other possible values 
of $k$ and $n$. The last expression goes to zero as $A\rightarrow \infty$ 
since $\delta > 1$. 

The remaining contribution to the probability 
(\ref{probball}) is then
\begin {eqnarray} \nonumber
&&\sum_{k=1}^{S+T} Z_N^{-1} W_0 \sum_{\substack{N_1+\ldots+N_{S+T} = 
N-|\tau_0| + T + S \\ N_n \leq A, ~n \neq k}} ~\prod_{i=1}^S Z_{N_i}^{(R)} 
\prod_{j=S+1}^{S+T}Z_{N_j} \\ \nonumber
&& \xrightarrow[N\rightarrow\infty]{} W_0 \zeta_0^{|\tau_0| - S - T} 
\Bigg(S (1-m+\zeta_0 g_R'(Z_0))\left(\sum_{n=1}^A Z_n^{(R)} 
\zeta_0^n\right)^{S-1}\left(\sum_{n=1}^A Z_n \zeta_0^n\right)^{T} \\ \nonumber 
&& \quad \quad \quad \quad \quad \quad \quad \quad \quad \quad \quad 
\quad + ~T \left(\sum_{n=1}^A Z_n^{(R)} \zeta_0^n\right)^{S}
\left(\sum_{n=1}^A Z_n \zeta_0^n\right)^{T-1}\Bigg) \\ \label{nuprob}
&& \xrightarrow[A\rightarrow\infty]{} W_0 \zeta_0^{|\tau_0| - S - T} 
\left(S(1-m+\zeta_0 g_R'(Z_0))Z_{0,R}^{S-1}Z_0^T+TZ_{0,R}^{S}Z_0^{T-1}\right).
\end {eqnarray}
This completes the proof of the existence of $\nu$. 
We now characterize $\nu$ separately for the cases $m=1$ and $m<1$.

\textit{The case $m=1$:} Let $\tau_1$ be a finite tree which has a vertex $s$ of degree one at a distance $R$ from the root. Let $A_R(\tau_1)$ be the set of trees which have $\tau_1$ as a subtree and the property that if the subrees attached to $s$ which do not contain the root are removed one obtains $\tau_1$, see Fig.~\ref{f:gen}. It is clear that $A_R(\tau_1)$ can be written as a finite union of pairwise disjoint balls as in $(\ref{ball})$. Therefore, by summing (\ref{nuprob}) over those balls we get
\begin {equation} \label{step2} 
\nu(A_R(\tau_1)) = W_1 \zeta_0^{|\tau_1| - 1} 
\end {equation}
where
\begin {equation} \label{W1}
W_1 = \prod_{v \in V(\tau_1) \setminus \{r,s\}} w_{\sigma_v}.
\end {equation}
Note that Equation (\ref{step2}) has the same form as (\ref{nuprob}) with $S=0$ and $T=1$.
 Now define $A_R$ as the union of $A_R(\tau_1)$ over all trees $\tau_1$ with the above properties. The sets $A(\tau_1)$ and $A(\tau'_1)$ are disjoint if $\tau_1 \neq \tau'_1$ and therefore by summing (\ref{step2}) over $\tau_1$ we find that
\begin {eqnarray} \nonumber
 \nu(A_R) &=& \sum_{\tau_1} \Big(\prod_{v \in V(\tau_1) \setminus \{r,s\}} w_{\sigma_v}\Big) \zeta_0^{|\tau_1|-1} = \Big(\zeta_0 \sum_{i=0}^\infty (i+1)w_{i+2} Z_0^i\Big)^{R-1} \\
&=& (\zeta_0 g'(Z_0))^{R-1} = m^{R-1} = 1
\end {eqnarray}
for all $R$. Therefore, by taking $R$ to infinity one finds 
that $\nu$ is concentrated on trees with exactly one spine with 
finite outgrowths. The distribution of the outgrowths follows 
from (\ref{step2}) and  (\ref{W1}).

\begin {figure} [!t]
\centerline{{\scalebox{0.7}{\includegraphics{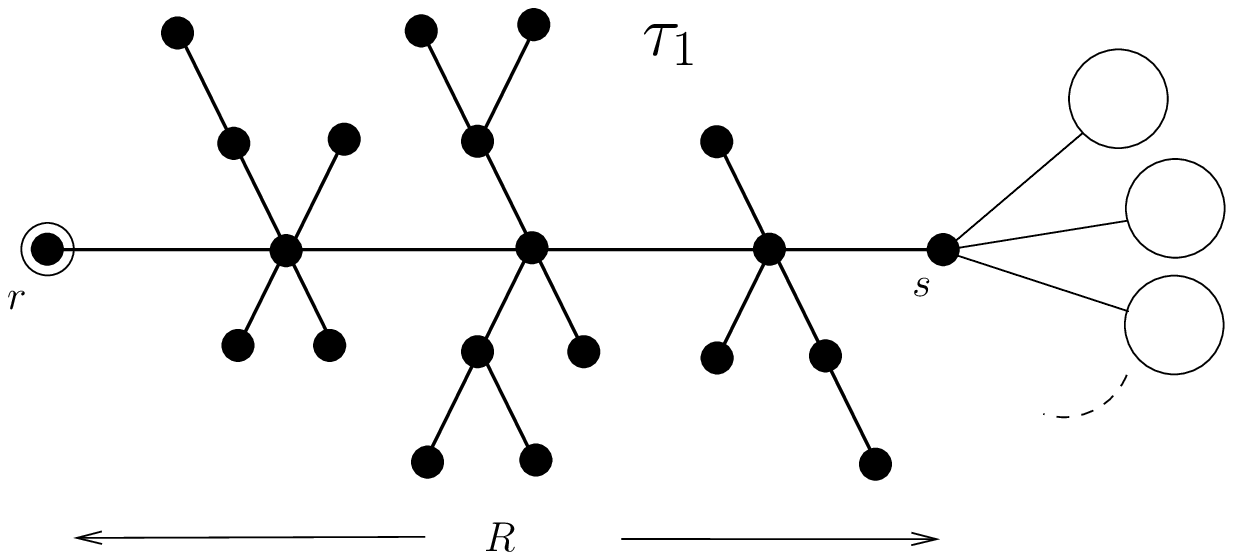}}}}
\caption{An illustration of the set $A_R(\tau_1)$.}\label{f:gen}
\end{figure}

\textit{The case $m<1$}: Let $\tau_2$ be a finite tree which has a vertex $t$ of degree $R$ at a distance $\ell$ from the root. Let $A_{R,\ell}(\tau_2)$ be the set of all trees which have $\tau_2$ as a subtree and the property that if the subtrees attached to $t$ in the $R$--th, $R+1$--st, $\ldots$ position clockwise from $(r,t)$ are removed one obtains $\tau_2$, see Fig.~\ref{f:ng}.  Summing  
(\ref{nuprob}) as in the case $m=1$ one finds that
\begin {equation} \label{step3}
 \nu(A_{R,\ell}(\tau_2)) = W_2 \zeta_0^{|\tau_2|-1}\left(1-m+\frac{g_R'(1)}{g(1)}\right), 
\end {equation}
where
\begin {equation} \label{W2}
W_2 = \prod_{v \in V(\tau_2) \setminus \{r,t\}} w_{\sigma_v}.
\end {equation}
Note that Equation (\ref{step3}) resembles (\ref{nuprob}) with $S=1$, $T=0$. Now define $A_{R,\ell}$ as the union of $A_{R,\ell}(\tau_2)$ over all trees $\tau_2$ with the above properties. By summing (\ref{step3}) over $\tau_2$ we get
\begin {equation} \label{nuar}
 \nu(A_{R,\ell}) = \left(1-m+\frac{g_R'(1)}{g(1)}\right)m^{\ell-1}.
\end {equation}
The sets $A_{R,\ell}$ are decreasing in $R$ so taking $R$ to 
infinity in (\ref{nuar}) one finds, by the monotone convergence 
theorem, that the probability of exactly one vertex having an 
infinite degree and being at a distance $\ell$ from the root 
is $(1-m)m^{\ell-1}$. Summing this over $\ell$ gives 1 which 
shows that the measure is concentrated on trees with exactly one 
vertex of infinite degree. The distribution of the outgrowths 
follows from (\ref{step3}) and (\ref{W2}).  \qed
\begin {figure} [!t]
\centerline{{\scalebox{0.7}{\includegraphics{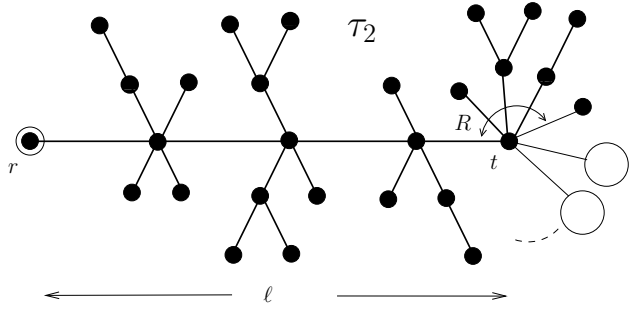}}}}
\caption{An illustration of the set $A_{R,\ell}(\tau_2)$.}\label{f:ng}
\end{figure} 
\end {proof}

\begin {theorem}
Theorem \ref{t:conv} applies to the generic and 
nongeneric, critical ensembles in Lemmas \ref{l:zgeneric} 
and \ref{ngcrit} and the nongeneric, subcritical 
ensembles defined by (\ref{bw}) and (\ref{ngcondition}).
\end {theorem}
\begin {proof}
This follows from Lemmas \ref{l:ZRN}, \ref{l:zgeneric} 
and \ref{ngcrit} since (\ref{asymZ}) holds with
\begin {equation}
 \delta = \left\{\begin{array}{ll}
3/2 & \quad \quad \text{generic, and nongeneric critical 
with $g''(1)<\infty$} \\
 \beta/(\beta-1) & \quad \quad \text{nongeneric critical with 
 $2 < \beta < 3$}\\
\beta & \quad \quad \text{nongeneric subcritical.}
 \end{array} \right.
\end {equation}
\qed
\end {proof}

The final result of this section 
concerns the size of the large vertex, in finite trees, which 
arises in the nongeneric, subcritical phase.
\begin {theorem}
Consider the nongeneric branching weights defined by (\ref{bw}) 
and (\ref{ngcondition}). Let $C_{N,\epsilon}$ be the event 
that a tree in $\Gamma_N$ has exactly one vertex of maximal
degree $\sigma_\mathrm{max}$ and $(1-m-\epsilon)N \leq \sigma_\mathrm{max} 
\leq (1-m+\epsilon)N$. For any $\epsilon,\delta > 0$ there 
exists an $N_0\in\mathbb{N}$ such that
\begin {equation}
 \nu_N\left(C_{N,\epsilon}\right) > 1-\delta
\end {equation}
 for all $N\geq N_0$.
\end {theorem}
\begin {proof}
This follows directly from the estimates (\ref{lastline}), 
(\ref{abovethird}) and (\ref{abovefourth}--\ref{belowsecond4}). \qed
\end {proof}
\section{The spectral dimension of subcritical trees}
In this section we will calculate the so called annealed spectral dimension of the nongeneric subcritical trees. A simple random walk on a graph $G$ is a sequence of nearest neighbour vertices, $\omega$, together with a probability weight
\begin {equation}
 \prod_{t=0}^{|\omega|-1} \left(\sigma_{\omega_t}\right)^{-1}
\end {equation}
where $\omega_t$ denotes the $(t+1)$-st vertex of $\omega$ and $|\omega|$ is the number of vertices in $\omega$. The random walk is a process where at time $t$ a walker, located at $\omega_t$, moves to one of its neighbours with probabilities $\left(\sigma_{\omega_t}\right)^{-1}$.

 Let $p_G(t)$ be the probability that a simple random walk which begins at the root in $G$, is located at the root at time $t$. The spectral dimension of the graph $G$ is defined as $d_s$ provided that
\begin {equation} \label{specdef}
 p_G(t) \asymp t^{-d_s/2}
\end {equation}
where we write $f(t)\asymp t^{-\gamma}$ if 
\begin {equation}
 \lim_{t\rightarrow\infty}\frac{\log{(f(t))}}{\log(t)} = -\gamma.
\end {equation}
If $p_G(t)$ falls off faster than any power of $t$ then we say that $d_s = \infty$.
The definition of $d_s$ is only useful on infinite graphs since on finite graphs, the return probability is asymptotically a positive constant. It is straightforward to verify that the spectral dimension of a connected, locally finite graph is independent of the choice of a root. 
The spectral dimension of the $d$--dimensional hyper--cubic lattice $\mathbb{Z}^d$ is $d_s = d$ in which case it agrees with our usual notion of dimension. For general graphs the spectral dimension need not be an integer and furthermore it might not exist.


For an infinite random graph $(\mathcal{G},\nu)$, where $\nu$ is a probability distribution on some set of graphs $\mathcal{G}$, one can define the spectral dimension in different ways. First of all the graphs can have, $\nu$ almost surely, a spectral dimension $d_s$ defined as above. 
Secondly, we define the {\it annealed spectral dimension} as $\bar{d}_s$ provided that
\begin {equation}
\langle p_G(t)\rangle_{\nu} \asymp t^{-\bar{d}_s/2}
\end {equation}
where $\langle \cdot \rangle_{\nu}$ denotes expectation value with respect to $\nu$. These definitions need not agree and we will see an example where $\bar{d}_s$ exists and is finite, whereas $d_s$ is almost surely infinite. For a discussion of the spectral dimension of some random graph ensembles, see \cite{durhuus:2005fq,durhuus:2007,jonsson:2008}. 

The Hausdorff dimension of a graph $G$ is defined in terms of how the volume of a graph ball $B_R(G)$ centered on the root scales with large $R$. The Hausdorff dimension is defined as $d_H$ if
\begin {equation}
 |B_R(G)| \asymp R^{d_H}.
\end {equation}
Similarly the {\it annealed Hausdorff dimension} is defined as $\bar{d}_H$ provided that
\begin {equation}
\langle |B_R(G)|\rangle_{\nu} \asymp R^{\bar{d}_H}.
\end {equation}
The spectral and Hausdorff dimensions do not agree in general. 

The Hausdorff dimension of subcritical trees is almost surely infinite and the annealed Hausdorff dimension is infinite. This follows from the fact that a vertex of infinite degree is almost surely at a finite distance from the root and that its expected distance from the root is finite. It is clear that the spectral dimension is almost surely infinite since a random walk will eventually hit the vertex of infinite degree and thereafter almost surely never return to the root. However, it turns out that the annealed spectral dimension is finite and takes the same values as in the case of subcritical caterpillars \cite{jonsson:2009}. The main result of this section is the following theorem.
\begin {theorem} \label{annspecdim}
For any $\beta > 2$ the annealed spectral dimension  of the subcritical trees defined by (\ref{bw}) and (\ref{ngcondition}) is
\begin {equation}
 \bar{d}_s = 2(\beta-1)
\end {equation}
provided it exists.
\end {theorem}

The return probabilities which we study to prove the above theorem, are most conveniently analysed through their generating functions. For a rooted tree $T$ define
\begin {equation}
 Q_{T}(x) = \sum_{t=0}^{\infty} p_T(t) (1-x)^{t/2}.
\end {equation}
The generating function variable $x$ is defined in this way for convenience in later calculations. Note that since $T$ is a tree 
only integer exponents appear on $1-x$. Let $p^1_T(t)$ be the probability that a random walk which leaves the root at time zero returns to the root for the first time after $t$ steps. Define the generating function
\begin {equation}
 P_{T}(x) = \sum_{t=0}^\infty p^1_T(t) (1-x)^{t/2}.
\end {equation}
By decomposing a walk which returns to the root into the first return walk, the second return walk etc.~we find the relation
\begin {equation} \label{relation}
 Q_T(x) = \sum_{n=0}^{\infty}(P_T(x))^n = \frac{1}{1-P_T(x)}.
\end {equation}
Let $n$ be the smallest nonnegative integer for which $Q^{(n)}_T(x)$, the $n$--th derivative of $Q(x)$, diverges as $x \longrightarrow 0$. If
\begin {equation}
(-1)^nQ_T^{(n)}(x)\asymp x^{-\alpha} 
\end {equation}
for some $\alpha\in [0,1)$ then clearly 
\begin {equation} \label{specalpha}
d_s=2(1-\alpha+n),
\end {equation}
if $d_s$ exists. For random graphs, the same relation holds between the singular behaviour of $\langle Q^{(n)}_T \rangle_\nu$ as $x\longrightarrow 0$ and the annealed spectral dimension. 
We will prove Theorem \ref{annspecdim} by establishing separately a lower bound and an upper bound on $\bar{d}_s$.  

\subsection{A lower bound on $\bar{d}_s$}

We first present a formula for the $n$--th derivative of a composite function (see e.g.~\cite{andrews:1998}) which will be used repeatedly.

\begin {lemma} (Fa\`{a} di Bruno's formula)
 If $f$ and $g$ are $n$ times differentiable functions then
\begin {equation}\label{faadibruno}
\frac{d^n}{dx^n}f(g(x)) = \sum_{\sum_{i=1}^n i q_i = n} \frac{n!}{q_1! q_2!\cdots q_n!}f^{(q_1+\ldots+q_n)}(g(x))\prod_{j=1}^{n} \left(\frac{g^{(j)}(x)}{j!}\right)^{q_j}.
\end {equation}
\end {lemma}
The following lemma will be needed to obtain the lower bound on $\bar{d}_s$.
\begin{lemma} \label{l:expfinite}
Let $\mu$ be a subcritical GW measure on $\Gamma'$ corresponding to the offspring probabilities (\ref{GWprob}). For any $n < \beta - 1$ and any nonnegative integers $\theta_1,\ldots,\theta_k$, $k\leq n$ such that $\theta_k \neq 0$ and  $\sum_{a=1}^{k} a\theta_a \leq n$ it holds that
\begin {equation} \label{expfinite}
 \left\langle\prod_{a=1}^k \left((-1)^{a}P^{(a)}_{T}(x)\right)^{\theta_a}\right\rangle_{\mu} < \infty
\end {equation}
for all $x\in[0,1]$.
\end{lemma}
\begin {proof}
The result is obvious for $x>0$ since the coefficients of $P_T(x)$ are smaller than one. First, take a fixed finite tree $T$ with root of degree one. Denote the degree of the nearest neighbour of the root by $N$ and the finite trees attached to that vertex by $T_1,\ldots,T_{N-1}$. Then from~\cite{durhuus:2007} we have the recursion
\begin{equation}
P_T(x) = \frac{1-x}{S_T(x)} 
\end{equation}
where
\begin {equation}
 S_T(x) = N - \sum_{i=1}^{N-1}P_{T_i}(x).
\end {equation}
Note that $S_T(x) \geq 1$, since $P_{T_i}(x)\leq 1$ for all $i$. By Fa\`{a} di Bruno's formula (with $f(x) = 1/x$, $g(x) = S_T(x)$) and throwing away negative powers of $S_T(x)$ we find that
\begin {eqnarray} \nonumber
  \frac{(-1)^{b}P_T^{(b)}(x)}{b!} &\leq& \sum_{\sum_{i=1}^b iq_i=b} \binom{q_1+\cdots+q_b}{q_1,\ldots,q_{b}}\prod_{j=1}^{b}\left(\frac{(-1)^{j+1}S_T^{(j)}(x)}{j!}\right)^{q_j}  \\ \nonumber && + \sum_{\sum_{i=1}^{b-1} iq_i=b-1} \binom{q_1+\cdots+q_{b-1}}{q_1,\ldots,q_{b-1}}\prod_{j=1}^{b-1}\left(\frac{(-1)^{j+1}S_T^{(j)}(x)}{j!}\right)^{q_j} \\ \label{fadibruno}
\end {eqnarray}
where $\binom{q_1+\cdots+q_b}{q_1,\ldots,q_{b}}$ is the multinomial coefficient. Looking at the product from the first sum we find that 
\begin {equation} \label{ppppp}
 \prod_{j=1}^{b}\left(\frac{(-1)^{j+1}S_T^{(j)}(x)}{j!}\right)^{q_j} = \prod_{j=1}^b \sum_{p_1+\cdots+p_{N-1}=q_j}\binom{q_j}{p_1,\ldots,p_{N-1}} \prod_{i=1}^{N-1}\left(\frac{(-1)^{j}P_{T_i}^{(j)}(x)}{j!}\right)^{p_i}.
\end {equation}
Expanding the above products and keeping track of the factors in each term which depend on the same outgrowth $T_i$, $i=1,\ldots,N-1$ we find that they are of the form
\begin{equation} \label{typicalterm}
C_i \prod_{j=1}^b  \left(\frac{(-1)^{j}P_{T_i}^{(j)}(x)}{j!}\right)^{\alpha_j}
\end{equation}
where $\sum_{j=1}^b j \alpha_j \leq b$ and $C_i$ is a number independent of $T_i$ (the terms in the latter sum in (\ref{fadibruno}) are of the same form, if $b$ is replaced by $b-1$). The equality $\sum_{j=1}^b j \alpha_j = b$ holds only when $p_i = \alpha_j = q_j$ in which case $p_a = 0$ if $a\neq i$ and $C_i=1$. The total contribution from such terms in (\ref{ppppp}) is therefore
\begin {equation}
 \sum_{i=1}^{N-1}\prod_{j=1}^b  \left(\frac{(-1)^{j}P_{T_i}^{(j)}(x)}{j!}\right)^{q_j}.
\end {equation}

Now choose numbers $\theta_1,\ldots,\theta_k$ such that $\theta_k \neq 0$ and $\sum_{a=1}^{k} a\theta_a \leq n$. Define $\Theta = \sum_{a=1}^k a \theta_a$. The following product of (\ref{fadibruno}) over $b$ has an upper bound
\begin {eqnarray*}
&& \prod_{b=1}^k \left(\frac{(-1)^{b}P_T^{(b)}(x)}{b!}\right)^{\theta_b} \leq \sum_{i=1}^{N-1}\prod_{b=1}^k\left(\frac{(-1)^jP_{T_i}^{(b)}(x)}{b!}\right)^{\theta_b} \\ && ~
 + C \sum_{M=1}^{\Theta} \sum_{\alpha(M)}\sum_{1\leq i_1<i_2<\cdots<i_{M}\leq N-1} \prod_{p=1}^{M}\prod_{b=1}^k  \left(\frac{(-1)^{b}P_{T_{i_p}}^{(b)}(x)}{b!}\right)^{\alpha_{b,i_p}}+C
\end {eqnarray*}
 where $\sum_{\alpha(M)}$ is a sum over nonnegative integers $\alpha_{b,i_p}$ which satisfy either
\begin {equation} \label{either}
\text{(i)}~ \sum_{b=1}^k b \alpha_{b,i_p} < \Theta \quad\quad \text{or} \quad \quad \text{(ii)}~ \sum_{b=1}^{k-1}b\alpha_{b,i_p} = \Theta
\end {equation}
 and $C$ is a number which only depends on $k$ and $(\theta_1,\ldots,\theta_k)$. Taking the $\mu$ expectation value of the above inequality and using the fact that the subtrees $T_i$, $i=1,\ldots,N-1$ are identically and independently distributed and distributed as $T$ itself, yields
\begin {eqnarray} \nonumber
 && \left\langle\prod_{b=1}^k \left(\frac{(-1)^{b}P_T^{(b)}(x)}{b!}\right)^{\theta_b}\right\rangle_{\mu} \leq  \\ 
 && \frac{C}{(1-m)g(1)} \sum_{M=1}^{\Theta} \sum_{\alpha(M)}\frac{g^{(M)}(1)}{M!} \prod_{p=1}^{M} \left\langle\prod_{b=1}^k  \left(\frac{(-1)^{b}P_{T_{i_p}}^{(b)}(x)}{b!}\right)^{\alpha_{b,p}}\right\rangle_{\mu}+\frac{C}{1-m}. \nonumber \\ \label{finiterec}
\end {eqnarray}
Note, that $M \leq \Theta \leq n < \beta -1$ and thus $g^{(M)}(1) < \infty$.  Therefore, for $x>0$, the right hand side of (\ref{finiterec}) is finite. To show that the left hand side is finite at $x = 0$ we proceed by induction on the sequences $(\theta_1,\theta_2,\ldots,\theta_k)$. We define a partial ordering on the set of such sequences in the following way (see also Fig.~\ref{f:trepun}).  Sequences $(\theta_1,\ldots,\theta_k)$ and $(\theta'_1,\ldots,\theta'_\ell)$ obey $(\theta'_1,\ldots,\theta'_\ell) < (\theta_1,\ldots,\theta_k)$ if and only if
\begin {equation*}
\text{(i)} ~ \sum_{i=1}^{\ell} i \theta'_i < \sum_{i=1}^{k} i\theta_i \quad \quad \text{or} \quad \quad \text{(ii)} ~ \sum_{i=1}^{\ell} i \theta'_i = \sum_{i=1}^{k} i\theta_i \quad \text{and} \quad \ell < k.
\end {equation*}
For the smallest values, $k=1$ and $\Theta = 1$, we find with the same calculations as above that
\begin {equation}
 \left\langle -P'_T(x)\right\rangle_{\mu} \leq \frac{1}{1-m}. 
\end {equation}
Next assume that (\ref{expfinite}) holds for for all sequences $(\theta'_1,\theta'_2,\ldots,\theta'_{k'})$ which are less than a given sequence $(\theta_1,\theta_2,\ldots,\theta_k)$ with $k,\Theta \leq n$. Then, by (\ref{either}), all the terms on the right hand side of (\ref{finiterec}) are finite and therefore the left hand side is finite for all $x\in[0,1]$. This shows that (\ref{expfinite}) holds for the sequence $(\theta_1,\theta_2,\ldots,\theta_k)$.   \qed
\end {proof}
\begin {figure} [!t]
\centerline{{\scalebox{0.49}{\includegraphics{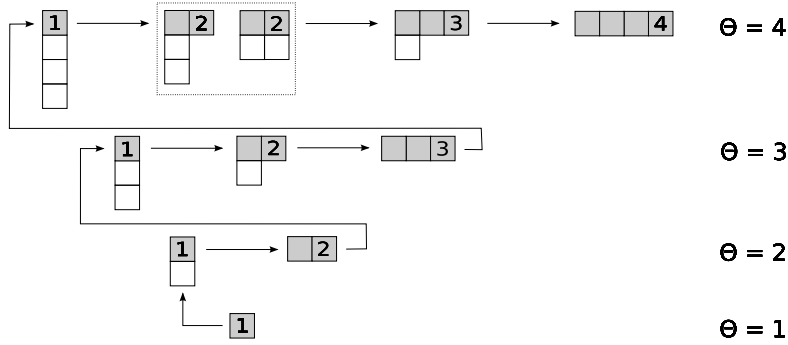}}}}
\caption{A sequence $(\theta_1,\theta_2,\ldots,\theta_k)$ is represented by a Young tableau where $\theta_i$ represents the number of rows of size $i$.  The size of a tableau is $\Theta$ and the number of elements in the top row (grey boxes) is the value of $k$. The tableaux are  first ordered by $\Theta$ and then by $k$ if possible. Tableaux with the same values of $\Theta$ and $k$ are incomparable.  } \label{f:trepun}
\end{figure}

Let $\nu$ be the measure corresponding to nongeneric subcritical trees as characterized in Theorem \ref{t:conv}. To find a lower bound on $\bar{d}_s$ with respect to $\nu$ we study an upper bound on a suitable derivative of the $\nu$-average return probability generating function. Let $M_{\ell}$ be a linear graph of length $\ell$ with the root, $r$, at one end and a vertex of infinite degree, $t$, on the other end. Let $B_{\ell,k}$ be the set of trees with graph distance $\ell$ between $r$ and $t$ and such that at least one vertex on the path $(r,t)$ has degree $k$ and all the other vertices have degree no greater than $k$ (with the exception of $t$ of course). Define 
\begin{equation}
 \langle \cdot \rangle_{\nu,\tau\in A} = \nu(A)^{-1}\sum_{\tau \in A} \nu(\tau) \left(\cdot\right)
\end{equation}
as the expectation value with respect to $\nu$ conditioned on the event $A$ and define
\begin {equation}
 \phi(k) = \sum_{i+j=k-2}\phi(i,j).
\end {equation}
We can write
\begin {equation} \label{frgf}
 \left\langle Q_\tau(x)\right\rangle_{\nu} =  \sum_{\ell=1}^\infty \psi(\ell) \sum_{k=2}^{\infty} c(k,\ell) \langle Q_\tau(x) \rangle_{\nu,\tau\in B_{\ell,k}} 
\end {equation}
where
\begin {equation}
 c(k,\ell) = \left(\sum_{i=2}^{k}\phi(i)\right)^{\ell-1} - \left(\sum_{i=2}^{k-1}\phi(i)\right)^{\ell-1}.
\end {equation}

 In a tree in $B_{\ell,k}$, denote the vertices on the path $(r,t)$ strictly between $r$ and $t$ by $s_1,s_2,\ldots,s_{\ell-1}$. Denote the outgrowths attached to $s_i$ by $T(s_i)$, where $i=1,\ldots,\ell-1$ and denote the $j$-th outgrowth from $s_i$ by $T_j(s_i)$ where $j=1,\ldots,\sigma_{s_i}-2$, see Fig.~\ref{f:blk}.  The first return probability generating function for $T(s_i)$ (viewing $s_i$ as the root) can be written in terms of the first return probability generating functions for $T_j(s_i)$ in the following way 

\begin{figure} [!t]
\centerline{\scalebox{0.67}{\includegraphics{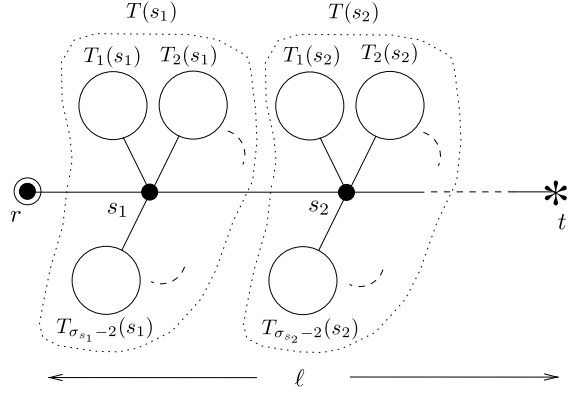}}}
\caption{A tree from $B_{\ell,k}$.} \label{f:blk}
\end{figure}

\begin {equation} \label{pgfrel}
 P_{T(s_i)}(x) = \frac{1}{\sigma_{s_i}-2} \sum_{j=1}^{\sigma_{s_i}-2}P_{T_j(s_i)}(x).
\end {equation}
Now take a $\tau\in B_{\ell,k}$. We can write
\begin {equation}
 Q_\tau(x) = \sum_{\substack{\omega:~r\rightarrow r\\ \text{on } M_\ell}} K_\tau(x,\omega) W_\omega(M_\ell) (1-x)^{|\omega|/2}
\end {equation}
where
\begin {equation}
 K_\tau(x,\omega) = \prod_{\substack{t=1 \\ \omega_t\in\{s_1,\ldots,s_{\ell-1}\}}}^{|\omega|-1} \frac{2}{2+(\sigma_{\omega_t}-2)(1-P_{T(\omega_t)}(x))}
\end {equation}
and
\begin {equation}
 W_\omega(M_\ell) = \prod_{t=0}^{|\omega|-1} (\sigma_{\omega_t}({M_\ell}))^{-1},
\end {equation}
see \cite{jonsson:2008}. Choose $n$ such that $n+1 < \beta \leq n+2$. Differentiating $n$ times we get

\begin {equation}
 \frac{(-1)^n Q^{(n)}_\tau(x)}{n!} = \sum_{n_1+n_2=n}\sum_{\substack{\omega:~r\rightarrow r\\ \text{on } M_\ell}} W_\omega(M_\ell)\frac{(-1)^{n_1}K^{(n_1)}_\tau(x,\omega)}{n_1!} \frac{(-1)^{n_2}}{n_2!}\frac{d^{n_2}}{dx^{n_2}}(1-x)^{|\omega|/2}.
\end {equation}

Let $\omega$ be a random walk and denote the maximal subsequence of $\omega$ which consists only of the vertices $s_1,\ldots,s_{\ell-1}$ by $\omega'$. Then
\begin {eqnarray*}
\frac{(-1)^b K^{(b)}_\tau(x,\omega) }{b!}=  \sum_{n_1+\cdots+n_{|\omega'|}=b} \prod_{t=1}^{|\omega'|}  \frac{(-1)^{n_{t}}}{n_{t}!}\frac{d^{n_{t}}}{dx^{n_{t}}}\left(\frac{2}{2+(\sigma_{\omega'_t}-2)(1-P_{T(\omega'_t)}(x))}\right).
\end {eqnarray*}
By Fa\`{a} di Bruno's formula we get
\begin {eqnarray} \nonumber
&& \frac{(-1)^p}{p!}\frac{d^p}{dx^p}\left(\frac{2}{2+(\sigma_{\omega'_t}-2)(1-P_{T(\omega'_t)}(x))}\right) = \\ \nonumber
&& \frac{2}{2+(\sigma_{\omega'_t}-2)(1-P_{T(\omega'_t)}(x))}  \sum_{q_1+2q_2+\cdots +pq_p = p} \binom{q_1+\cdots +q_p}{q_1,\ldots,q_p}  \\ \nonumber 
&& \times~ \Big(\underbrace{\frac{2(\sigma_{\omega'_t}-2)}{2+(\sigma_{\omega'_t}-2)(1-P_{T(\omega'_t)}(x))}}_{(\ast)}\Big)^{q_1+\cdots+q_p}   \prod_{a=1}^p \left(\frac{(-1)^{a}P^{(a)}_{T(\omega'_t)}(x)}{a!}\right)^{q_a}. \\ \label{star}
\end {eqnarray}
Now, $P_{T(\omega'_t)}(x) \leq 1-x$.  Also note that the quantitiy ($\ast$) in (\ref{star}) is increasing in $\sigma_{s_i}$ and since $\sigma_{s_i} \leq k$ for $i=1,\ldots,\ell-1$ we find that 
\begin {equation} \label{ast}
(\ast) \leq \frac{2(k-2)}{2+(k-2)x}.
\end {equation}
Observe that $\dfrac{2(k-2)}{2+(k-2)x} \leq 1$ for $k=2,~3$ and that $\dfrac{2(k-2)}{2+(k-2)x} \geq 1$ for $k \geq 4$. Finally, note that $\binom{q_1+\cdots +q_p}{q_1,\ldots,q_p} \leq p^p$. Combining these results and using (\ref{pgfrel}) we get the upper bound

\begin {eqnarray} \nonumber
&&\frac{(-1)^b K^{(b)}_\tau(x,\omega) }{b!} \leq  b^b\left(\frac{2(k-2)}{2+(k-2)x}\right)^{(1-\delta_{k,2})(1-\delta_{k,3})b} \\ && \nonumber \sum_{n_1+\cdots+n_{|\omega'|}=b}  \prod_{t=1}^{|\omega'|}  \sum_{q_1+2q_2+\cdots +{n_t}q_{n_t} = n_t} \prod_{a=1}^{n_t} \frac{1}{(\sigma_{\omega'_t}-2)^{q_a}} \\ \nonumber && \times ~ \sum_{p_1+\cdots+p_{\sigma_{\omega'_t}-2} = q_a} \binom{q_a}{p_1,\ldots,p_{\sigma_{\omega'_t}-2}}\prod_{j=1}^{\sigma_{\omega'_t}-2}\left(\frac{(-1)^{a}P^{(a)}_{T_j(\omega'_t)}(x)}{a!}\right)^{p_j}. \\ \label{dK}
\end {eqnarray}
Expanding the above products and keeping track of the factors in each term which depend on the same outgrowth $T_j(s_i)$, $i=1,\ldots,\ell-1$, $j=1,\ldots,\sigma_{s_j}-2$, we find that they are of the form
\begin{equation} \label {terms}
C_{ij} \prod_{a=1}^n \left((-1)^{a}P^{(a)}_{T_j(s_i)}(x)\right)^{\theta_a}
\end{equation}
where $\sum_{a=1}^{n}a\theta_a \leq n$ and $C_{ij}$ is independent of $T_j(s_i)$. By Lemma \ref{l:expfinite}, the expected value of (\ref{terms}) over the outgrowths $T_j(s_i)$ is finite, and since the total number of terms on the right hand side of (\ref{dK}) is a polynomial in $|\omega'|$ we find that
\begin {equation}
\left\langle(-1)^b K^{(b)}_\tau(x,\omega)\right\rangle_{\nu,\tau \in B_{\ell,k}} \leq H(|\omega|) \left(\frac{2(k-2)}{2+(k-2)x}\right)^{(1-\delta_{k,2})(1-\delta_{k,3})b}
\end {equation}
where $H(|\omega|)$ is a polynomial with positive coefficients. From this inequality and the fact that $(-1)^iQ^{(i)}_{M_\ell}(0)$ is a polynomial in $\ell$ of degree $2i+1$, it follows that
\begin {equation}
\langle (-1)^n Q^{(n)}_\tau(x) \rangle_{\nu, \tau \in B_{\ell,k}} \leq \sum_{i=0}^{n}S_i(\ell)\left(\frac{2(k-2)}{2+(k-2)x}\right)^{(1-\delta_{k,2})(1-\delta_{k,3})i}
\end {equation}
where $S_i(\ell)$, $i=0,\ldots,n$ are polynomials with positive coefficients. 
Noting that $c(k,\ell) \leq \phi(k)(\ell-1)$ we get from (\ref{frgf}) that
\begin {equation}
 \langle (-1)^n Q^{(n)}_\tau(x) \rangle_{\nu} \leq  \sum_{i=0}^{n} \sum_{\ell=1}^\infty S_i(\ell)\psi(\ell) (\ell-1) \sum_{k=2}^{\infty} \phi(k) \left(\frac{2(k-2)}{2+(k-2)x}\right)^{(1-\delta_{k,2})(1-\delta_{k,3})i}.
\end {equation}
The sum over $\ell$ is convergent since $\psi$ falls off exponentially and the sum over $k$ can be estimated by an integral yielding
\begin {equation} \label{intest}
 \langle (-1)^n Q^{(n)}_\tau(x) \rangle_{\nu} \leq  C x^{\beta-n-2} \int_{x}^\infty \frac{y^{n+1-\beta}}{(2+y)^{n+1}}dy
\end {equation}
where $C$ is a constant. If $\beta < n+2$ the last integral is convergent when $x \rightarrow 0$ but if $\beta = n+2$ it diverges logarithmically. In both cases we get the lower bound $\bar{d}_s \geq 2(1-\beta)$, provided $\bar{d}_s$ exists.
\qed
\subsection{An upper bound on $\bar{d}_s$}
To find an upper bound on $\bar{d}_s$ we study a lower bound on a suitable derivative of the average return probability generating function. The aim is to cut off the branches of the finite outgrowths from the path $(r,t)$ so that only single leaves are left. We then use monotonicity results from~\cite{jonsson:2008} to compare return probability generating functions.    As before we choose $n$ such that $n+1 < \beta \leq n+2$. We begin by differentiating (\ref{frgf}) $n$ times and throwing away every term in the sum over $\ell$ except the $\ell=2$ term
\begin {equation} \label{frgfII}
 \left\langle (-1)^{n}Q^{(n)}_\tau(x)\right\rangle_{\nu} \geq  (1-m)m \sum_{k=2}^{\infty} \phi(k)\left\langle (-1)^{n}Q^{(n)}_\tau(x)\right\rangle_{\nu,\tau\in B_{2,k}}. 
\end {equation}
Let $M_{2,k}$ be the graph constructed by attaching $k-2$ leaves to the vertex $s_1$ in $M_2$ defined in the previous section.  Take a tree $\tau \in B_{2,k}$. Denote the nearest neighbours of $s_1$, excluding $r$ and $t$, by $u_1,\ldots,u_{k-2}$. Denote the finite tree attached to $u_i$ by $U(u_i)$, $i=1,\ldots,k-2$, and view $u_i$ as its root, see Fig.~\ref{f:m2k}.
\begin{figure} [!b]
\centerline{\scalebox{0.67}{\includegraphics{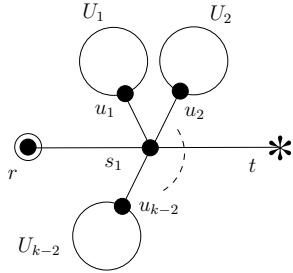}}}
\caption{A graph $\tau \in B_{2,k}$.} \label{f:m2k}
\end{figure}
 We can write
\begin {equation}
Q_\tau(x) = \sum_{\substack{\omega:~r\rightarrow r\\ \text{on } M_{2,k}}} F_\tau(x,\omega) W_\omega(M_{2,k})( (1-x)^{|\omega|/2}
\end {equation}
where
\begin {equation}
 F_\tau(x,\omega) = \prod_{\substack{t=1 \\ \omega_t\in\{u_1,\ldots,u_{k-2}\}}}^{|\omega|-1} \frac{1}{1+(\sigma_{\omega_t}-1)(1-P_{U(\omega_t)}(x))}.
\end {equation}
Define
\begin {equation}
H(x) =  \sum_{\substack{\omega:~r\rightarrow r\\ \text{on } M_{2,k}}}\ \left\langle F_\tau(x,\omega)\right\rangle_{\nu,\tau\in B_{2,k}} W_\omega(M_\ell) \frac{d^{n-1}}{dx^{n}}(1-x)^{|\omega|/2}.
\end {equation}
Differentiating once we easily find that
\begin {equation}
(-1)^nH'(x) \leq \left\langle (-1)^{n}Q^{(n)}_\tau(x)\right\rangle_{\nu,\tau\in B_{2,k}}                                      
\end {equation}
and using the methods of~\cite[Section 4]{jonsson:2008} we find that there exists a sequence $\xi_i$ converging to zero as $i \to \infty$ on which
\begin {equation}
 (-1)^nQ^{(n)}_{M_{2,k}}(\xi_i) \leq (-1)^nH'(\xi_i).
\end {equation}
Using the relation (\ref{relation}) one can show that $(-1)^nQ_\tau(x) \geq (-1)^nP_\tau(x)$ for any $\tau$. Thus, we finally have
\begin {equation} \label{festimate}
 \left\langle (-1)^{n}Q^{(n)}_\tau(\xi_i)\right\rangle_{\nu} \geq  (1-m)m \sum_{k=2}^{\infty} \phi(k) (-1)^nP^{(n)}_{M_{2,k}}(\xi_i)
\end {equation}
on a sequence $\xi_i$ converging to zero. In \cite{jonsson:2009} it is shown that
\begin {equation}
 P^{(n)}_{M_{2,k}}(x) = (-1)^n \frac{n!(k-1)^{n-1}k}{(2+(k-2)x)^{n+1}}
\end {equation}
and therefore the sum over $k$ in (\ref{festimate}) can be estimated from the below  by the same integral as in (\ref{intest}) up to a multiplicative constant. This proves that $\bar{d}_s \leq 2(\beta-1)$ provided $\bar{d}_s$ exists.
 \qed
\section {Conclusions}
We have studied an equilibrium statistical mechanical model of trees and shown that it has two phases, an elongated phase and a condensed phase. We have proven convergence of the Gibbs measures in both phases and on the critical line separating them. The main result is a rigorous proof of the emergence of a single vertex of infinite degree in the condensed phase. The phenomenon of condensation  appears in more general models of graphs \cite{agishtein:1992,ambjorn:1992} and it would be interesting to prove analogous results in those cases. 

In the generic phase the annealed Hausdorff dimension is $\bar{d}_H = 2$ and the annealed spectral dimension is $\bar{d_s} = 4/3$, see \cite{durhuus:2007}. The proof of this result relies only on the fact that the infinite volume measure is concentrated on the set of trees with exactly one spine having finite critical Galton--Watson outgrowths and that $g''(1) < \infty$. Therefore, it follows from Theorem \ref{t:conv} that $\bar{d}_H = 2$ and $\bar{d}_s = 4/3$ on the critical line when $g''(1) < \infty$. 


It remains an open problem to calculate the dimension of trees on the critical line when $g''(1) = \infty$. It is easy to see that the annealed Hausdorff dimension is infinite in this case since the expected value of the degree of any vertex on the spine is infinite. However, we expect from the analogous case of caterpillars \cite{jonsson:2009} and on the basis of scaling arguments \cite{burda:2001dx,correia:1997gf} that
\begin {equation} \label{critdim}
 d_H = \frac{\beta-1}{\beta - 2} \quad \quad \quad \text{and} \quad \quad \quad d_s = \frac{2(\beta-1)}{2\beta-3}
\end {equation}
holds almost surely, where $2 < \beta \leq 3$. Note that by Theorem \ref{t:conv}, the infinite volume measure is still concentrated on the set of trees with exactly one spine having critical Galton--Watson outgrowths. Therefore, a possible way to prove (\ref{critdim}) is to follow the arguments in \cite{durhuus:2007}, but taking into account the different behaviour of critical Galton--Watson processes having $g''(1) = \infty$. Some results on such Galton--Watson processes can be found in \cite{Slack:1968}. 

In the condensed phase the Hausdorff and spectral dimension are almost surely infinite due to the infinite degree vertex. The same applies to the annealed Hausdorff dimension. However, the annealed spectral dimension takes the values $\bar{d}_s = 2(\beta-1)$ where $\beta > 2$. This is different from the value $\bar{d}_s = 2$ which was obtained in \cite{correia:1997gf} using scaling arguments. The reason is that the scaling ansatz used in \cite{correia:1997gf} is apparently not valid when a vertex of infinite degree appears.
\medskip \\ \noindent
{ \bf Acknowledgement.} The work of SÖS was supported by the Eimskip fund at the University of Iceland. This work was partly supported by Marie Curie grant MRTN-CT-2004-005616 and the University of Iceland Research Fund. We are indebted to B.~Durhuus, G. Miermont and W. Westra for helpful discussions. 

\bibliographystyle{siam}
\bibliography{n1.bbl}

%
\end{document}